\documentclass[conference,compsoc]{IEEEtran}
\IEEEoverridecommandlockouts

\usepackage[utf8]{inputenc} 
\usepackage[T1]{fontenc}    
     \usepackage{natbib}

\usepackage{tabularx}
\usepackage{thmtools} 
\usepackage{thm-restate}
\usepackage{lipsum}
\usepackage{xcolor}         
\label{key}
\usepackage{comment}
\usepackage{amsmath}
\usepackage{epsfig}
\usepackage{wrapfig}
\usepackage{subfigure}
\usepackage{multirow}
\usepackage{hyperref}
\usepackage{graphicx}
\usepackage{caption}
\usepackage{array}
\usepackage{bbm}
\usepackage{color}
\usepackage{xcolor}
\usepackage{enumerate}
\usepackage{enumitem}
\setlist{leftmargin=10mm}
\usepackage{mathtools}
\usepackage{amsmath,amssymb,amsthm,bm}
\usepackage{thmtools}
\usepackage{amsfonts,graphicx}
\usepackage{mathrsfs}
\usepackage{algorithm,algorithmic}\newtheorem{remark-star}{Remark}
\newtheorem{remark-star-1}{Remark}

\usepackage{amsthm}
\newtheorem{thm}{Theorem}
\usepackage{thmtools,thm-restate}

\newtheorem{definition}{Definition}

\newtheorem{lemma}{Lemma}
\newtheorem{remark}{Remark}
\newtheorem{proposition}{Proposition}
\usepackage{booktabs}

\usepackage{lipsum}

\usepackage{etoolbox}

\newcommand{\definecalletter}[1]{%
  \expandafter\newcommand\csname c#1\endcsname{\mathcal{#1}}%
}
\forcsvlist{\definecalletter}{A,B,C,D,E,F,G,H,I,J,K,L,M,N,O,P,Q,R,S,T,U,V,W,X,Y,Z}

\newcommand{\E}{\mathbb{E}}
\newcommand{\R}{\mathbb{R}}

\newcommand{\softmax}{\mathrm{softmax}}

\usepackage{xparse}
\newcommand{\Ex}{\E}
\newcommand{\adv}{\mathbf{adv}}
\newcommand{\supp}{\mathrm{supp}}
\newcommand{\ball}{\mathrm{Ball}}
\NewDocumentCommand{\fDPi}{O{$(\epsilon,\delta)$}}{%
  #1-$\mathrm{DP}^\Psi_i$~%
}
\NewDocumentCommand{\fDPr}{O{$(\epsilon,\delta)$}}{%
  #1-$\mathrm{DP}^\Psi_r$~%
}
\NewDocumentCommand{\DPr}{O{$(\epsilon,\delta)$}}{%
  #1-$\mathrm{DP}_r$~%
}
\NewDocumentCommand{\DPi}{O{$(\epsilon,\delta)$}}{%
  #1-$\mathrm{DP}_i$~%
}

\newcommand{\Sim}{\mathrm{sim}}

\newcommand{\poisson}{\mathrm{poisson}}

\newcommand{\set}[1]{\{#1\}}

\usepackage[T1]{fontenc}
\usepackage{titlesec}
\titleformat{\paragraph}[runin]
{\bfseries\scshape}{\theparagraph}{1em}{}
\newcommand{\myparagraph}[1]{\par\vspace{1em}\textbf{#1:}}
\usepackage{cite}
\usepackage{amsmath,amssymb,amsfonts}
\usepackage{algorithmic}
\usepackage{graphicx}
\usepackage{textcomp}
\usepackage{xcolor}
\def\BibTeX{{\rm B\kern-.05em{\sc i\kern-.025em b}\kern-.08em
    T\kern-.1667em\lower.7ex\hbox{E}\kern-.125emX}}
\begin{document}

\title{Machine Learning with Privacy for Protected Attributes
%
}

\author{\IEEEauthorblockN{Saeed Mahloujifar}
\IEEEauthorblockA{\textit{FAIR at Meta} \\
saeedm@meta.com}
\and
\IEEEauthorblockN{Chuan Guo}
\IEEEauthorblockA{FAIR at Meta \\
chuanguo@meta.com}
\and
\IEEEauthorblockN{G. Edward Suh}
\IEEEauthorblockA{\textit{NVIDIA / Cornell University} \\
esuh@nvidia.com}
\and
\IEEEauthorblockN{Kamalika Chaudhuri}
\IEEEauthorblockA{\textit{FAIR at Meta} \\
kamalika@meta.com}
}

\maketitle

\begin{abstract} Differential privacy (DP) has become the standard for private data analysis. Certain machine learning applications only require privacy protection for specific protected attributes. Using naive variants of differential privacy in such use cases can result in unnecessary degradation of utility. In this work, we refine the definition of DP to create a more general and flexible framework that we call feature differential privacy (FDP). Our definition is simulation-based and allows for both addition/removal and replacement variants of privacy, and can handle arbitrary and adaptive separation of protected and non-protected features.

We prove the properties of FDP, such as adaptive composition, and demonstrate its implications for limiting attribute inference attacks. We also propose a modification of the standard DP-SGD algorithm that satisfies FDP while leveraging desirable properties such as amplification via sub-sampling. We apply our framework to various machine learning tasks and show that it can significantly improve the utility of DP-trained models when public features are available. For example, we train diffusion models on the AFHQ dataset of animal faces and observe a drastic improvement in FID compared to DP, from 286.7 to 101.9 at $\epsilon=8$, assuming that the blurred version of a training image is available as a public feature.

Overall, our work provides a new approach to private data analysis that can help reduce the utility cost of DP while still providing strong privacy guarantees. \end{abstract}

\section{Introduction} Differential privacy (DP) has emerged as the gold standard for privacy protection. The typical formulation of DP considers adjacent datasets that differ by one sample and requires the output of a DP algorithm to be distributionally indistinguishable when executed on adjacent datasets. However, there are applications, particularly in the context of machine learning, where the unit of privacy is not necessarily the entire sample.

Consider a simple example: training a generative model that outputs images of cars while ensuring that no real license plate in the training set is memorized and reproduced. One plausible solution to achieve this is to crop (or blur) all license plates from the training images; however, this will lead to the model generating unrealistic cars (low utility). Another potential solution is to train a generative model with a DP guarantee (e.g., using DP-SGD \citep{song2013stochastic, abadi2016deep}); this can ensure that the probability of reproducing a license plate is small via bounds on data reconstruction~\citep{guo2022bounding, hayes2024bounding}. However, training the model with DP could come with a major utility cost because the privacy definition is insensitive to which part of the image we try to protect.

This separation of sensitive and non-sensitive attributes occurs in different machine learning applications and across different modalities: For images, only certain parts, such as license plates, credit card numbers, or the face of a person, may be considered private; for text data, the most privacy-sensitive part could be personally identifiable information (PII) such as names, addresses, phone numbers, \emph{etc.}; for tabular datasets, we might only consider a few columns to be private while the rest could be public. When such separation of private information is possible, it may be desirable to refine the definition of DP beyond the granularity of a single sample. Doing so allows the model to memorize non-sensitive aspects of a training sample while providing rigorous protection for sensitive information.\ 
\begin{figure}[t]
    \centering
\includegraphics[width=0.9\linewidth]{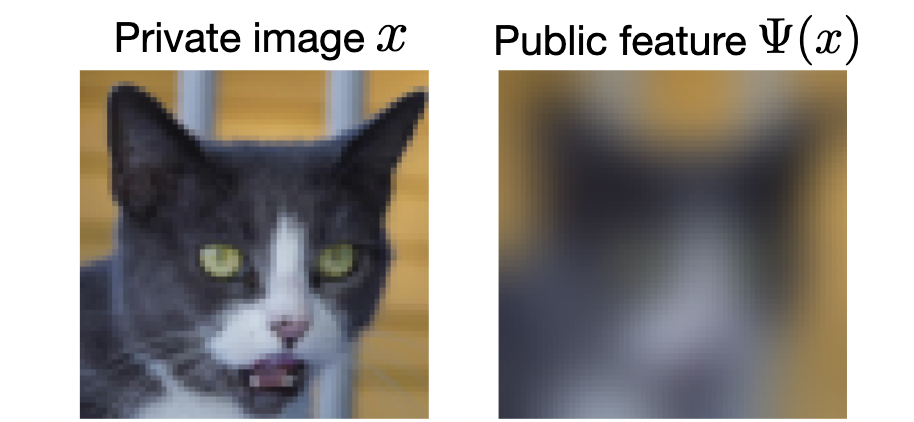}
\caption{Private image and its public variant. \label{fig:public_vs_private}}
\includegraphics[width=0.9\linewidth]{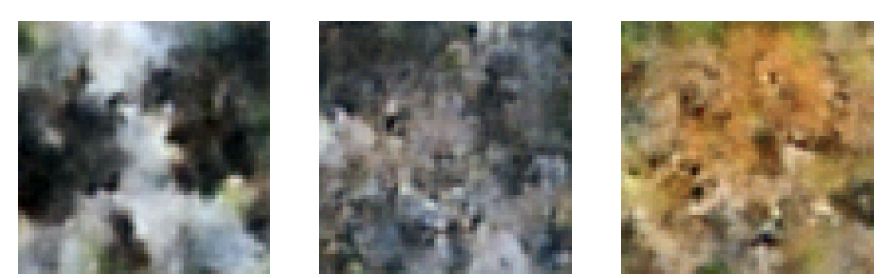}
\caption{Generation results for DP training. \label{fig:afhq_dp_samples}}
\includegraphics[width=0.9\linewidth]
{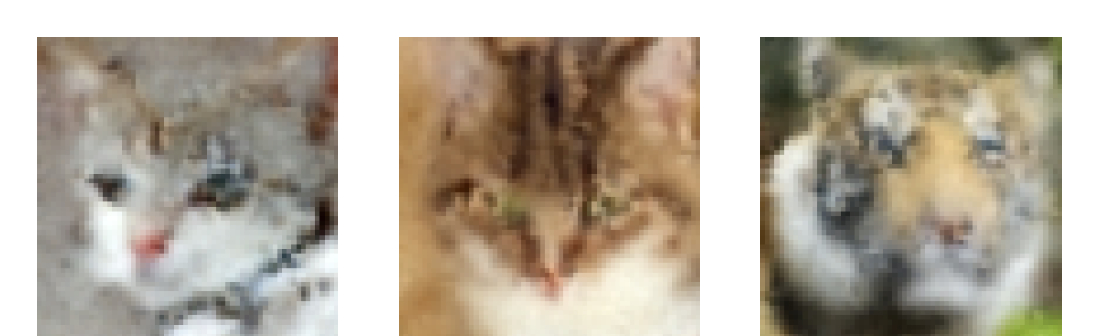}
    \caption{Generation results for DP training while leaking the public feature.  \label{fig:afhq_fdp_samples}}
\end{figure}
 Motivated by this discussion, we propose and analyze a relaxation of DP called \emph{feature differential privacy} (or \emph{feature-DP}). In a nutshell, feature DP ensures that replacing a license plate with another license plate does not significantly change the distribution of an output. Since this is a relaxation of DP, we can potentially achieve much better utility while preventing the reproduction of license plates to the same degree as DP.

To formulate the separation of private and public features, our definition leverages a public feature map $\Psi$. Intuitively, for a given sample $x$, $\Psi(x)$ should encode any information about $x$ that is publicly accessible, while all other information in $x$ is considered private. For example, for a private image $x$ in \autoref{fig:public_vs_private} (left), the public feature map $\Psi$ can be the blurring operation, with the blurred image $\Psi(x)$ in \autoref{fig:public_vs_private} (right) assumed to be publicly accessible.

Allowing public features to be a function of the input makes our definition of feature DP highly flexible, enabling it to be used in almost any scenario with partial privacy protection. This definition also enjoys desirable properties of DP, such as composition and post-processing. We emphasize that we are not the first to identify the importance of partial privacy protection. Previous work has explored partial protection of privacy in various contexts. However, our definition has desirable properties that make it more general and suitable for many scenarios. In the Related Work section, we compare our approach with previous definitions that aim to protect privacy for certain features.

Furthermore, we propose a variant of the quintessential DP learning algorithm—DP-SGD~\citep{abadi2016deep}—that adapts to feature-DP and can benefit from amplification by sub-sampling. We empirically demonstrate the advantage of feature-DP on two popular machine learning tasks—classification and diffusion model training—and show that it improves the privacy-utility trade-off compared to DP.

\myparagraph{Our contributions} \begin{enumerate}[leftmargin=*] \item We propose and define \emph{feature differential privacy}. Our definition is highly flexible, allowing one to express different forms of private information for various data modalities, and satisfies desirable properties such as composition and post-processing. We further analyze the privacy implications of feature-DP by connecting it to attribute inference attacks.
\item We propose a variant of DP-SGD that specializes in feature-DP. Notably, it is possible to show that a straightforward adaptation does not benefit from privacy amplification by sub-sampling. Our variant solves this issue by sampling two separate batches of data in each iteration.

\item We analyze our variant of DP-SGD for convex optimization and show convergence results. We obtain concrete bounds for logistic regression that outperform the utility of standard DP-SGD while protecting the private features with the same $(\epsilon,\delta)$-parameters. Notably, we are also able to answer an open question raised in \citep{ghazi2021deep} regarding label differential privacy (which is a special case of feature DP). Specifically, we achieve excess risk bounds that improve as the size of the dataset grows. This is precisely a consequence of our framework and the fact that it can leverage sub-sampling amplification.

\item We perform experiments on various ML datasets to showcase the benefit of feature-DP in improving the privacy-utility trade-off. These experiments cover classification and image generation tasks on tabular and vision datasets. Notably, we train private diffusion models on the AFHQ dataset of animal faces. \autoref{fig:afhq_dp_samples} shows samples generated by a diffusion model trained with a DP guarantee of $\epsilon=8$, attaining an FID of $286.7$. In comparison, \autoref{fig:afhq_fdp_samples} shows samples generated by a diffusion model trained with a feature-DP guarantee of $\epsilon=8$, which shows substantial improvement in visual quality and attains a much lower FID of $101.9$.
\end{enumerate}

\section{Related Work}
In this section, we discuss several related works that bear similarity to ours. Several studies have explored notions of privacy that allow privacy to be defined in a partial manner. Below, we summarize some of these notions and highlight the novelty of our approach relative to them.

\begin{itemize} \item \textbf{Generalized notions of differential privacy:} Previous works, such as Metric DP~\citep{chatzikokolakis2013broadening} and Pufferfish privacy~\citep{kifer2012private}, have considered generalizations of differential privacy that can be used for modeling various privacy settings. These notions are almost too general for our use case, and interpreting them as a way to protect the privacy of certain features is somewhat challenging. Moreover, in contrast to our definition, these definitions can only be instantiated for add/remove notions of privacy and are only defined for the traditional notion of approximate DP, which is parameterized by $(\epsilon, \delta)$ (unlike our definition, which is based on trade-off functions and resembles $f$-Differential privacy).

\item \textbf{Partial DP:} Perhaps a more related notion is Partial Differential Privacy, introduced in~\citep{ghazi2022algorithms}. This notion attempts to define privacy in a way that provides a granular notion of privacy based on how many attributes are protected. In particular, this notion considers settings where there are $d$ well-defined attributes. Then, the privacy protection will gradually increase as the number of sensitive attributes decreases. This notion could be relevant in settings where the sensitive attributes are not well-defined. However, in our setting, we are interested in scenarios where there is a clear set of sensitive attributes, and we want to protect their privacy. That is, there are certain attributes that we can completely leak and still remain private. It is worth noting that the notion of Partial DP is only applicable to datasets that are neighboring with respect to replacement and also does not have a well-defined $f$-DP variant.

\item \textbf{Selective DP:} Selective differential privacy~\citep{shi2021selective} considers privacy-preserving language modeling on data containing sensitive information such as personally identifiable information (PII). They define selective DP in a similar manner by using automated tools to detect sensitive tokens and protect only the sensitive tokens with a DP guarantee. Their DP guarantee ensures that replacing the PIIs with another set of PIIs will not significantly change the final model. Although this notion is very similar to ours in purpose, it is only defined for language models. Moreover, their privacy notion is restricted to the replacement notion of adjacency and does not benefit from amplification by sub-sampling.

\item \textbf{Tabular Data:} The work of~\citep{krichene2023private} considered a notion of private learning with public features for tabular data, where the data is assumed to contain categorical features with additional descriptors that are public. For example, in movie recommendation, the training samples can be all movies that a user is interested in, and the public descriptors for movies can contain information such as the director, actors, genre, \emph{etc.} This privacy notion is stronger than ours and is closer in nature to standard DP.

\item \textbf{Information Theoretic Privacy:} There are other information-theoretic notions of privacy that allow for the separation of sensitive and public attributes (see~\citep{mireshghallah2021not} and references therein). However, in this work, we adhere to differential privacy, as it has become the gold-standard notion of privacy, especially for machine learning applications.

\item \textbf{Label Differential Privacy:} Finally, we point out that label-differential privacy~\citep{ghazi2021deep, malek2021antipodes}, a privacy notion defined for classification with private labels, is a special case of our definition. Our notion of privacy enables one to define label-differential privacy by setting the label to be the only private attribute. 
\item \textbf{Distribution privacy and attribute privacy} We also note some bodies of work that try to protect the global properties of datasets.  Earlier work~\citep{ganju2018property, mahloujifar2022property} had shown that the distribution of certain properties of a dataset can be inferred by adversaries (e.g. the gender distribution in census data). Zhang et. al.~\citep{zhang2022attribute} defines attribute privacy as a way to formalize and protect privacy against these attacks. On the other hand, the work of Lin et. al.~\citep{lin2022distributional} considers distributional privacy, where the concern is about the role of data release in leaking the underlying distribution of the data. Both attribute privacy and distributional privacy are about the leakage from the underlying distribution and are not concerned with individual samples, unlike our definition.
\item \textbf{Inferential privacy} Ghosh et. al.~\citep{ghosh2016inferential} have defined the notion of inferential privacy. This notion is similar to differential privacy but, instead of assuming data points are independent, it allows for arbitrary correlation between the samples in the dataset. What makes this definition relevant to us is the way it can still be defined even when there are correlations between different data points. An interesting observation made by one of our anonymous reviewers was that one can leverage this definition in the local setting and define privacy for a specific sensitive feature. Since inferential privacy allows for correlation, this will be useful for feature privacy when different features could be correlated.
However, there are a few shortcomings with this definition. First, this definition would require an explicit set of attributes and then it would satisfy privacy for all of them, but in a granular way. So, it becomes similar to the work of \citep{ghazi2022algorithms} discussed above. Second, we are not aware of any optimization algorithm that satisfies this definition and leverages non-private attributes. Also, we should point out that this definition would only be useful for differing datasets with replacement, unlike our definition.
\end{itemize}
\section{Preliminaries}
In this section, we detail some of the preliminaries needed for differential privacy and its properties. We first begin by introducing some of our notations.
\myparagraph{Notations} We use calligraphic (e.g. $\mathcal{X}$) letters to denote sets.We use capital letters to denote random variables that are often supported on the calligraphic version of the letter (e.g. a random variable $X$ supported on $\mathcal{X}$). We use lowercase letter to denote members of a set (e.g. $x\in \mathcal{X})$. We use $x\sim X$ to denote the process of sampling a point $x$ from the random variable $X$. We use $\mathcal{X}^k$ to denote the space of all datasets of size $k$ with samples in $\mathcal{X}$. We use $\mathcal{X}^*$ to denote the union of all $\mathcal{X}^k$ for all integers $k\geq 0$.
\begin{definition} [Neighboring datasets]
    We call a pair of datasets $(\cS, \cS')$ $k$-neighboring with replacement, and denote it by $\cS\approx_k^r \cS$ iff $|\cS|=|\cS'|$ and $|\cS \Delta \cS'|=2k.$ (Here, $\Delta$ is the symmetric difference operation.) Similarly, we call them $k$-neighboring with insertion/deletion and denote it by $\cS\approx_k^{id} \cS'$ iff $|\cS\Delta \cS'| = k$.  When $k=1$, we simply say neighboring datasets instead of $1$-neighboring datasets and denote it with $\approx^r$ or $\approx^{id}$. 
\end{definition}
\begin{definition}[Differential Privacy \citep{Dwork2006}]
A randomized mechanism $M\colon \mathcal{X}^* \to \cO$ is $(\epsilon,\delta)$-DP with replacement (with insertion/deletion) if for all pairs of neighboring datasets $\cS\approx^r\cS'$ ($\cS\approx^{id}\cS')$ and all subsets of outcomes $\cT$ we have 
$$\Pr[M(\cS) \in \cT]\leq e^\epsilon \Pr[M(\cS')\in \cT] + \delta.$$
\end{definition}
Note that we have defined both forms of differential privacy (DP) — with replacement and with insertion/deletion — simultaneously. The privacy guarantee under replacement is strictly stronger, as it implies insertion/deletion privacy with the same parameters. Nevertheless, when reporting differential privacy parameters, it is common practice to use the insertion/deletion formulation.

Although this notion of DP is widely adapted , there is a finer version of DP named $f$-DP that we define next. This notion is very useful for ``privacy accounting'' purposes. This notion is based on trade-off functions defined below.
\begin{definition}[Trade-off function ]
    For random variables $X$ and $Y$ supported on a set $\cX$, the trade-off function between $X$ and $Y$ is function $f=T(X,Y)$ where $f$ is defined as
    $$f(\alpha) =  1- \sup_\cT\set{\Pr[Y\in \cT]: \Pr[X\in \cT]\leq \alpha}.\footnote{Note that $X, Y$ could be of an arbitrary dimension. For instance they can be multidimensional Gaussians. However, the range and domain of the trade-off function is always on $[0,1]$.}$$
\end{definition}
\begin{definition}[$f$-Differential Privacy~\citep{dong2019gaussian}] Let $f \colon [0,1] \to [0,1]$ be an arbitrary trade-off function.
A randomized mechanism $M\colon \mathcal{X}^* \to \cO$ is $f$-DP  with replacement, denoted by $f$-DP$_r$, (with insertion/deletion, denoted $f$-DP$_i$) if for all pairs of neighboring datasets $\cS\approx^r\cS'$ ($\cS\approx^{id}\cS')$ and all subsets of outcomes $\cT$ we have 
$$\Pr[M(\cS) \in \cT]\leq 1- f\Big(\Pr[M(\cS')\in \cT]\Big).$$
\end{definition}
Note that the traditional notion of DP is essentially a linear lower bound on the optimal $f$-DP curve for a mechanism. Although expressing the curve with only two parameters is nice in that it can be used as a benchmark for privacy, it is still better to work with the notion of $f$-DP for composition and privacy accounting. That is, when we want to compose two mechanisms $M_1$ and $M_2$, it is better to first obtain the tightest $f$-DP for each mechanism, compose the $f$-DP of the two mechanisms, and then convert the resulting $f$-DP curve to $(\epsilon,\delta)$-DP. 
Next, we discuss the details of composition amd conversion of $f$-DP to DP.s
\begin{thm}[Composition for $f$-DP~\citep{dong2019gaussian}]
Consider a series of  mechanisms $M_1,\dots, M_k$ where each $M_i$ is $f_i$-DP, with $f_i=T(X_i,Y_i)$. Then the adaptive composition  $$M=M_1 o M_2, o \dots o M_k$$ is $f$-DP, where 
$$f=T(X_1\times\dots\times X_k, Y_1\times \dots \times Y_k).$$ and $X\times X'$ is the product distribution of $X$ and $X'$.
\end{thm}

We also state a theorem on how to convert $f$-DP to $(\epsilon,\delta)$-DP.
\begin{thm}[$f$-DP to DP conversion~\citep{dong2019gaussian}]
If a mechanism is $f$-DP, then for any $\delta\in[0,1]$ it $(\epsilon,\delta)$-DP with $$\epsilon=\sup_{\alpha\in[0,1]} log(\frac{1-\delta -f(\alpha)}{\alpha}).$$
\end{thm}
We note that calculating the optimal $f$-DP for (sub-sampled) Gaussian mechanism and its composition has been the focus of recent body of work \citep{dong2019gaussian, gopi2021numerical,wang2022analytical,zhu2022optimal,wang2023randomized}. 
In this work, we will also use these results to calculate tight privacy bounds. 

\section{Feature Differential Privacy}\label{sec:feature_DP}
As inferred from its definition, DP aims to regulate an adversary's capability to distinguish whether a particular data-point was incorporated in the dataset or not~\citep{Dwork2006}. The objective of our work is to relax this definition and control the adversary's information access up to a certain level of permissible leakage. Specifically, rather than permitting the neighboring datasets $\cS$ and $\cS'$ to be chosen in the worst-case scenario, we necessitate the differing points to be identical for some of their features. This requirement imposes an increased challenge for the adversary to differentiate between $\cS$ and $\cS'$, but simultaneously obligates us to willingly expose those features. 
Definition \ref{def:fdp} below formalizes this notion.

\begin{definition} [Feature Differential Privacy]\label{def:fdp}
Let $\Psi: \cX \to \cF$ be a feature that maps a point $x\in \cX$ to a point $u\in \cF$ in the feature space, and let $f:[0,1]\to [0,1]$ a trade-off function. We say a mechanism $M$ is $f$ feature differential private with replacement, relative to $\Psi$, and denote it by \fDPr[$f$], if for all neighboring datasets $\cS\approx^{r}\cS’$  such that
\begin{align*}\cS\Delta \cS’=\set{x, x'}  \text{~and~}   \Psi(x)=\Psi(x’)
\end{align*}
and for all subsets $\cT$ of the range of the mechanism we have
$$\Pr[M(\cS)\in \cT]\leq 1-f(\Pr[M(\cS')\in \cT]).$$ 
Moreover, we say a mechanism is \fDPr iff it is \fDPr[$f$] for $f(x)=1-\delta - e^\epsilon\cdot x$.
\end{definition}

In this definition, $\Psi$ represents a part of the data points that can be disclosed. For instance, imagine a tabular dataset of individuals wherein the education level of each individual is anticipated to be publicly available. In this case, we can assign the feature $\Psi$ as the education level, and the definition of feature-DP necessitates that the adversary should be unable to confidently infer whether a data-point $x$ or $x'$ was used in the training set, given that $x$ and $x'$ have the same level of education, known to the adversary. As we increase the amount of information in $\Psi$, the feature DP definition becomes a weaker notion of privacy.
One can also view feature DP as a generalization of \emph{label differential privacy}~\citep{ghazi2021deep,malek2021antipodes, tang2022machine}, where the adversary is assumed to have access to all the features except for the label.

\myparagraph{Considerations for public leakage}  As we discussed, our definition of feature DP allows for the leakage of a feature $\Psi(x)$ for all examples $x$ in the dataset. This leakage might violate some protections that would be otherwise protected through traditional notion differential privacy. Consider the case of privacy for license plates where we allow everything except for the license plates in the image to be leaked. One limitation with this approach is that the public feature can be used for \textit{identification}. If one sees an image of Batmobile, they can likely identify the owner of the car without seeing the license plate. For this very reason, feature DP would not necessarily provide upper bounds on membership inference attacks. Rather, as we discuss in Section \ref{sec:feature_DP}, feature DP would prevent \textit{attribute inference} on the private features.

The second consideration is potential imperfections in the public features that might become problematic for privacy. For example, in the case of masking license plates, the masking will most likely be automated using a detector. If there are errors in the detector then the public feature might leak more information than intended. This issue is however not unique to feature-DP, as automated filters such PII redaction are already being used as tools for privacy protection. The quality of these filters needs to be assessed comprehensively to avoid privacy violations.  In any case, the choice of the public feature will be determined by the privacy requirements; it is not a design choice for us. Although choosing the public feature is a subtle task, our contribution is to realize privacy for any chosen public attribute. We do not recommend using the public features that we use in our empirical studies in the paper. 

\myparagraph{Insertion/deletion variant of feature DP} Definition \ref{def:fdp} applies to the replacement notion of adjacency. In recent years, the community has mostly focused on the notion of ``insertion/deletion'' when reporting the differential privacy parameters. At first glance it might seem like it is not possible to define a ``insertion/deletion'' for feature differential privacy because by leaking the public features, one would also leak the size of the dataset. In fact, to the best of our knowledge, previous work have not considered such a definition for the notion of label differential privacy either. However, it is possible to define the ``insertion/deletion'' variant of feature-DP (and label-DP as a special case) in a subtle way using a simulator.



\begin{definition} [Simulation based definition for insertion/deletion]
    Let $\Psi\colon \cX \to \cF$ be a feature and $f\colon [0,1]\to [0,1]$ be a trade-off function. We say a mechanism $M\colon\cX^* \to \cO$ is $f$ feature DP relative to $\Psi$, denoted by \fDPi[$f$], if there is a simulator $\Sim\colon \cX^*\times \cF \to \cO$ such that for all $x\in \cX$ and all neighboring datasets $\cS\approx^{id}\cS'=\cS\cup \set{x}$ and $\cT\subseteq \cO$ we have
\begin{align*}
&\small \Pr[M(\cS')\in \cT]\leq 1-f\Big(\Pr[\Sim(\cS,\Psi(x))\in \cT]\Big)\\
&\text{~~and~~}\\
&\Pr[\Sim(\cS,\Psi(x))\in \cT]\leq 1-f\Big(\Pr[M(\cS')\in \cT]\Big)
\end{align*}
We say $M$ is \fDPi when $f(x) = 1-\delta-e^\epsilon\cdot x$.
\end{definition}
\myparagraph{The role of simulator} Simulation-based definitions of security have a rich history in cryptography \citep{datta2005relationships}, and we are not the first to leverage this concept. Cryptographers use simulators to argue that an adversary cannot learn much about a random variable $X$ when given access to another related random variable $Y$. They demonstrate that for any adversary $A$ with access to $Y$, the adversary's output can be simulated without $Y$. This shows that adversary gains no extra information through $Y$. This notion is used in defining cryptographic primitives like Multi-Party Computation, Zero-Knowledge Proofs, and Program Obfuscation.

In our work, we apply this logic to differential privacy. The simulator attempts to simulate the output of a mechanism while knowing the entire training set expect for one sample for which it only knows some public features. We impose a restriction on the divergence between the mechanism's output and the simulator's output. It's important to note that the simulator is not part of the mechanism; it defines the mechanism's privacy. This approach allows us to define privacy for features akin to traditional DP's insertion/deletion privacy.

We now state two propositions that show how this notion relates to the earlier notion of feature DP with replacement and the traditional notion of DP.  
\begin{proposition}\label{prop:fdpi_to_FDPr}
    If a mechanism $M$ satisfies feature \fDPi[$f$], then it satisfies \fDPr[$(f \circ f)$] . Similarly, if $M$ satisfies \fDPr[$f$], then it also satisfies \fDPi[$f$].
\end{proposition}
It is also worth noting that if a mechanism satisfies DP with insertion/deletion, it also satisfies feature DP with insertion/deletion.
\begin{proposition}\label{prop:DP_to_FDP}
    If a mechanism is \DPi[$f$], then it is also \fDPi[$f$]  with respect to any feature $\Psi$.
\end{proposition}
In the rest of this section, we state important properties of the notion of feature DP. We first prove that composition could be done similar to DP. Then we state the implication of feature-DP for attribute inference. 

\subsection{Composition}
Now we prove composition for feature DP. Again, because of the subtleties of the definition, we need to prove the composition separately for the cases of replacement and insertion/deletion. For this, we directly work with the notion of $f$-DP. Although the proof of this composition is very similar to that of composition for DP, it still has some subtleties with the notion of simulator. As the composition of mechanisms could be adaptive, the simulator for each mechanism could depend on the state of the previous mechanisms. So we need to construct a simulator by adaptively composing the simulator for each mechanism and then argue about the closeness of distributions.

\begin{thm}[Composition for feature DP] \label{thm:FDP_composition} Let $\set{f_j=T(X_j,Y_j);j\in [k]}$ be a series of trade-off functions and $\set{M_j\colon \cX^* \to \cO; j\in[k]}$ a series of mechanisms where each $M_j$ is \fDPi[$f_j$], where $\Psi\colon \cX \to \cF$ is an arbitrary feature. Then the adaptive composition of $M_j$s is \fDPi[$f$] for $$f=T(X_1\times\dots\times X_k, Y_1\times\dots\times Y_k).$$
\end{thm}

\subsection{Feature DP limits attribute inference}
Differential privacy nicely maps to the threat model of membership inference attacks. It prevents an attacker from guessing whether or not a point $x$ was included in the dataset. This prevention holds even if the adversary has a non-uniform prior about the inclusion of the point $x$. That is, even if the adversary knows the point $x$ is in the training set with probability $90\%$, even before observing the output of the DP mechanism, she will still remain uncertain after observing the model, albeit with less uncertainty. This stems from the ``hypothesis testing'' interpretation of differential privacy.

In the following, we develop a threat model that corresponds to feature DP. This threat model bears similarity to the threat of attribute inference attacks studied in previous work \citep{jayaraman2022attribute,jia2018attriguard}.
To formulate our threat model, we consider an adversary who has a dataset $\cS$ in mind. The adversary also chooses a feature value $u\in \cF$, and also a distribution $X$ over points in $\cX$. Then, we sample a point from the conditional distribution, conditioned on the feature value, $x\sim X \mid  \Psi(X) = u$. We run the mechanism $\theta\gets M(S \cup \set{x})$, and reveal $\theta$ to the adversary. The adversary then needs to reconstruct $x$, based on $\theta$. Note that the adversary already has the knowledge that $\Psi(x)=u$. The following definition exactly measure the power of adversary in reconstructing the rest of $x$ (beyond what is known by $\Psi(x)=u$).

\begin{definition}
    Let $\Psi \colon \cX \to \cF$ be a feature, $X$ a distribution of points supported on $\cX$ and $u\in \cF$ a feature value. Let $M: X^* \to \Theta$ be a mechanism, and $d$ a distance metric and $\rho$ a threshold for the metric $d$. Let $A: \Theta \to X$ be an attribute inference attack, we define the success rate of this attack as:
    $$\adv(A, X, u, \cS, d, \rho)= \Pr_{\substack{x\sim [X\mid \Psi(X)=u]\\ \theta \sim M(S \cup \set{x})\\
    x' \sim A(\theta)}}[d(x', x)\leq \rho].$$
\end{definition}

Now, we aim at bounding this quantity for a mechanism that satisfies feature DP, for a feature $\Upsilon$. But before that, we need to define another notion, that quantifies how concentrated (according to metric $d$) the distribution $D\mid u$ is. 

\begin{definition}
We define $\ball(\cS, u, d, \rho)$ to be the maximum volume of a $d$-Ball of radius $\rho$ in the measure space of $X\mid u$. That is
$$\ball(X, u, d, \rho) = \sup_{x^* \in \supp(X)} \Pr_{x\sim X\mid u}[d(x,x^*)\leq \rho]$$
\end{definition}
Note that this quantity measures how well an adversary that does not even observe the output of the mechanism can reconstruct a point sampled from $X\mid u$. This notion essentially relates to entropy of distribution in a geometric sense. It is natural for any upper bound on the advantage of attribute inference to depend on this quantity as we can never make the upper bound to be smaller than this quantity. In other words, we can only try to minimize the blow-up in the probability of successful reconstruction, when the attacker can additionally observe the output of the mechanism. The following theorem does precisely that.  
\newpage
\begin{thm}\label{thm:att_inf}
Let $f\colon [0,1]\to[0,1]$ be a convex trade-off function and let $M$ be a mechanism that satisfies feature \fDPi[f] with respect to a feature $\Psi:\cX\to \cF$. Let $X$ be a data distribution and $u\in \cF$ be an arbitrary feature value. Then for all adversaries $A$ we have
$$\adv(A,X,u, \cS, d, \rho) \leq 1-f\Big(\ball(X,u, d,\rho)\Big).$$
\end{thm}
\myparagraph{Comparing feature DP with DP:} In our experiment section, we perform experiments on several datasets to obtain utility privacy trade-off for both feature DP and DP and compare them. However, we note that the comparison between the privacy parameters for feature DP and DP is only relevant if our goal is to protect attribute inference attacks on the private features. We know if a mechanism is DP, then it will be safe against attribute inference attacks depending on how small the privacy parameter is. Interestingly, as we show in Theorem \ref{thm:att_inf}, we get the same level of safety against attribute inference attacks when we have mechanisms that are Feature DP with the same privacy parameter. This is why we have plotted DP and Feature-DP together in our experiments. 

\section{Optimization with Feature DP}\label{sec:optimization} 

In this section we aim at developing a framework to perform private optimization while leveraging the public features.
Assuming a dataset $D$ and a loss function $\ell$, we target the resolution of the optimization problem where we want to find $\min_\theta \sum_{x\in \cS}\ell(x,\theta)$, with feature differential privacy with respect to a feature $\Psi$. DP-SGD, the differentially private variant of SGD, is an algorithm designed to resolve this optimization with DP assurances. A single iteration of DP-SGD involves computing the gradient of the loss function on a randomly chosen subset of data points, aggregating these to yield the gradient of the total loss function on the batch, and subsequently introducing noise to the aggregated gradient. Under the assumption that loss function is smooth with Lipschitz constant $\ell$, DP  for each optimization step can be attained. The random batch selection further amplifies the privacy of each step, utilizing the known results on privacy amplification by sub-sampling~\citep{balle2018privacy}. Concurrently, considering that the addition of noise and subsampling doesn't bias the optimization, one can also obtain proven guarantees for DP-SGD's convergence, under the right set of assumptions on the loss function \citep{bassily2019private, kifer2012private}. Our goal is to alter the process of DP-SGD, so that we can leverage the public features and achieve better privacy-utility trade-off. 

\myparagraph{Challenge 1--Disentangling the public and private signals} One challenging problem with leveraging the public features is that in DP-SGD, the signal for training is in the form of gradients. Due to the complexity of loss functions in most applications, we cannot disentangle the private and public signals. To solve this problem, we propose a framework that tries to separate the signal from the public features by introducing a secondary loss function that is defined over the public features and is designed to predict the true loss function. Intuitively, we hope the loss function $\ell'$ has the property that $\ell(x,\theta) \approx \ell'(\Psi(x),\theta)$.
We reformulate the minimization problem as follows:
$$\min_\theta \underbrace{\sum_{x\in \cS} \Big(\ell(x,\theta)-\ell'(\Psi(x),\theta)\Big)}_{\text{private loss}} + \underbrace{\sum_{x \in \cS} \ell'(\Psi(x), \theta)}_{\text{public loss}}.$$

Note that although this is the exact same optimization, we now have two separate gradients from the two terms, while the second term only depends on the public features and is public. Now, if we ensure that $\ell(x,\theta)- \ell'(\Psi(x),\theta)$ possesses a smaller Lipschitz constant than that of $\ell$ we can improve the privacy analysis without changing the optimization objective. Although it might appear that we can instantly attain better privacy after this trick, given that we have reduced the Lipschitz constant of the loss function applied to private features, we yet need to make another change.
\\
\myparagraph{Challenge 2--Feature DP is not amplified by sub-sampling} The second challenge is regarding sub-sampling amplification. Unfortunately, feature-DP in general does not benefit from amplification by sub-sampling.

\begin{proposition}\label{prop:sub-samp-FDP}
    There exists a mechanism $M$ where $M  o~\poisson_p$ (Performing Poisson sub-sampling with probability $p$ before applying the mechanism $M$) is tightly \fDPi[$(\epsilon, 0)$] for all $p\in[0,1]$.
\end{proposition}

This can unfortunately limit the use-cases of feature DP in private machine learning where random batch selection often plays an important role in achieving reasonable trade-off between accuracy, privacy, and efficiency. To address this challenge, we propose an alternate algorithm that employs two separate batches for the two losses. Our approach is detailed in Algorithm \ref{alg:feature-DP}. Using this alternative algorithm, we can still benefit from amplification by sub-sampling. In a  nutshell, we use two separate batches for the public and private loss functions. This avoids leakage of the random coins of Poisson sub-sampling and hence we can leverage that randomness in our privacy analysis.\\

\begin{algorithm}
\caption{Noisy SGD with Public Features}\label{alg:feature-DP}
\begin{algorithmic}[1]
\REQUIRE Public feature $\Psi$, Dataset $\cS$, Batch sizes $m$, $m'$, Learning rate schedule $\eta(t)$, standard deviation $\sigma$, Projection space $\mathcal{W} \in R^d$, Loss functions $\ell_{priv}, \ell_{pub}$, Number of iterations $T$
\STATE Initialize $\mathbf{w}_1 \in \mathcal{W}$
\FOR{$t = 1, \dots, T$}
    \STATE Sample a mini-batch $B^{priv}_t$ with Poisson sampling with probability $p=m/|\cS|$. 
    \STATE $g^{priv}_{t}~=~\frac{1}{m}\sum_{x\in B_t^{priv}}\nabla \ell_{priv}(\mathbf{w}_t; x)$
    \STATE Sample a second mini-batch $B^{pub}_t$ of size $m'$ uniformly at random.
    \STATE $g_t^{pub}=~\frac{1}{m'}\sum_{x\in B_t^{pub}}\nabla \ell_{pub}(\mathbf{w}_t; \Psi(x)).$
    \STATE Let $g_t~=~g_t^{pub} + g_t^{priv} + \mathcal{N}(0, \sigma^2)$
    \STATE Update $\mathbf{w}_{t+1} = \mathbf{w}_t - \eta(t)\cdot g_t$
    \STATE Project $\mathbf{w}_{t+1}$ into the set $\mathcal{W}$: $\mathbf{w}_{t+1} = \text{Proj}_{\mathcal{W}}(\mathbf{w}_{t+1})$
\ENDFOR
\STATE \textbf{return} $\mathsf{aggregate}(\mathbf{w}_1,\dots,\mathbf{w}_{T+1})$
\end{algorithmic}
\end{algorithm}

\myparagraph{Privacy analysis}
Now we present a formal analysis of our algorithm's privacy, followed by an evaluation of its utility in the convex case.

\begin{thm}\label{thm:privacy}[Privacy analysis]
Assume the private loss function in Algorithm \ref{alg:feature-DP} is $\tau$-Lipschitz. 
Then, Algorithm \ref{alg:feature-DP} is \fDPi[f] for $$f=T\Big(\cN(0,\sigma)^T, \big((1-p)\cN(0,\sigma) + p \cN(\tau, \sigma)\big)^T\Big).$$
In particular, setting
$$\sigma = c\frac{\tau m}{\epsilon\cdot n} \sqrt{T\log(\frac{1}{\delta}) \log(\frac{T}{\delta})},$$
the mechanism is  \fDPi.
\end{thm}
Note that standard deviation of the noise is independent of the lipschitz constant of $\ell_{pub}$ and this enables obtaining a better utility for the same privacy.\\
\begin{remark}[Clipping instead of lipschitsness]
Although our privacy analysis in Theorem \ref{thm:privacy} is based on the lipschitz constant of the private loss function, we only need the private gradient to be bounded for the privacy analysis to hold. To privatize non-lipschitz loss functions, we can simply clip the gradient to an arbitrary threshold and use that instead of the lipschitz constant to obtain the privacy guarantee.
\end{remark}
\myparagraph{Utility analysis}  To illustrate the potential utility gain from the public features, we show how to bound the excess risk of the algorithm for convex and Lipschitz loss functions.

\begin{thm}[Excess empirical risk]\label{thm:utility} Assume $\ell=\ell_{priv}+\ell_{pub}$ is convex, and $\tau'$-Lipschitz. Let $M=\max_{w\in \mathcal{W}} \mid\!\mid w\mid\!\mid$. Let $l^*=\min_{w\in \cW} \sum_{x\in \cS} \ell(x)$. Then, setting $m'=|\cS|$ in Algorithm \ref{alg:feature-DP}, 
and using learning rate $\eta(t)=\frac{c}{\sqrt{t}}$
we have
\begin{align*}\E[\sum_{x\in \cS} \ell(\mathbf{w}_{T+1},x)] - l^*
\leq (\frac{M^2}{c} + c\tau'^2 + cd\sigma^2)\frac{2+\log(T)}{\sqrt{T}}.
\end{align*}
Moreover, if $\ell$ is $\lambda$-strongly convex, with $\eta(t)=\frac{1}{t\lambda}$ we have 
$$\E[\sum_{x\in \cS}\ell(\mathbf{w}_{T+1},x)] - l^*
\leq 17(\tau'^2 + d\sigma^2)\frac{1+\log(T)}{\lambda T}.$$
\end{thm}

Note that this excess empirical risk is smaller than what one can obtain from the analysis of DP-SGD (see Lemma 3.3 in \citep{bassily2019private}) because our noise is proportional to $\tau$ (the lipschitz function of the private portion of the loss function), instead of $\tau'$ (the lipschitz contat of the entire loss). We now illustrate the utility gain of our analysis for logistic regression in the specific case of Label-DP.
\\
\myparagraph{Case Study: Logistic Regression}
Suppose we have a  datasets that is of the form $\cS=\set{(x_i,y_i); i\in[n]}$ where $x_i \in [0,1]^{d}$ and $y_i \in  \set{0,1}^k$ is the one-hot encoding of the label. We aim at minimizing the logistic loss
    $$\min_{\mathbf{w}\in \R^{d\times k}}\ell(\mathbf{w}, \cS)=\sum_{i\in [n]} -y_i \log(\softmax(\mathbf{w}^T x_i))$$
    while preserving the privacy of labels. Namely, the leaked feature $\Psi$ is defined by $\Psi(x,y)=x$. (This setting is also referred to as Label DP~\citep{ghazi2021deep}). Let us define the public loss function $\ell_{pub}(\mathbf{w},x)$ so that we have $\nabla \ell_{pub} = x\cdot p^T$, where $p$ is the probability vector, namely, 
$p=\softmax(x\cdot \mathbf{w})$ . Since this public loss does not depend on the label $y$, it can indeed be treated as public. Then, we set the 
private loss function $\ell_{priv}(\mathbf{w},x,y) = \ell(\mathbf{w},x,y) - \ell_{pub}(\mathbf{w},x)$. Consequently, we can say that $\ell=\ell_{priv}+\ell_{pub}$.
Using these definitions of private and public losses in our Algorithm \ref{alg:feature-DP}, we can derive the concrete bounds for this scenario.
Assume that the input space $X=\set{x\in R^d, \lvert\!\lvert x\rvert\!\rvert \leq 1}$. Observe that the gradient of $\ell_{priv}$ equals $-xy^T$.
Therefore, the loss function $\ell_{priv}$ is $1$-Lipschitz. This is while the gradient of $\ell$ is equal to $-x(y^T - p^T)$, which is only $\sqrt{2}$-Lipschitz. 
Comparing this bound for the bound one would obtain from analyzing DP-SGD, we are saving a factor of $\sqrt{2}$ in the Lipschitz constant, therefore, we can use a smaller $\sigma$ to obtain the same $\epsilon$ and $\delta$. Therefore, we save a factor two in the dominant $cd\sigma^2$ term in the excess risk analysis.  \\
\myparagraph{Comparison with Ghazi et al., 2021}
The work of \citep{ghazi2021deep} presents an excess risk bound for label-DP under convex functions. In their work, they use an algorithm where they use batch of size 1 at each round, but instead of leverging amplification by sub-sampling, they use the fact that the sub-space of all possible gradients has a low rank (the rank is equal to the number of labels). Hence, they can only add noise to the basis of the sub-space which can reduce the magnitude of the noise significantly. Although this is a nice way to avoid privacy amplification by sub-sampling, it remains sub-optimal for the cases that a huge dataset is available. 
In their work they raise an open question on if one can leverage sub-sampling for label-DP. Our work answers their open question affirmatively. For the logistic regression setting, their empirical excess risk is bounded by $$O(\frac{k}{\epsilon^2} \log(\frac{1}{\delta})).$$
This is while our excess empirical risk bounded by $$O(\frac{d}{\epsilon^2\cdot |\cS|^2} \log(\frac{1}{\delta})).$$ The division by $|\cS|^2$ in our case is due to our ability to leverage amplification by sub-sampling, a capability they lack. Particularly, as long as the number of examples $|\cS|$ is greater than $\sqrt{d/k}$, our algorithm yields a lower excess risk than theirs.

\begin{remark}
We note that our Algorithm \ref{alg:feature-DP} can be combined with an arbitrary number of public pre-training steps. That is, we can perform as many steps with public loss function as needed, and that would not affect the privacy analysis. Indeed, in some of our experiments (see Section \ref{sec:experiments}) we find that this public pre-training can lead to significant improvements in the results. 
\end{remark}
\section{Experimental Results}\label{sec:experiments}

\subsection{Datasets and associated public features} We use Purchase100, Criteo, LSUN, and AFHQ. Below, we specify what features we considered public for each dataset.

\begin{itemize}[noitemsep, topsep=0pt, leftmargin=*]
    \item Purchase100 is a tabular dataset that contain 600 real valued features and a label from a set of 100 possible labels. In our experiments, we consider the case where a set of 100 features (out of 600) are considered public. We also consider the label to be public.
    \item Criteo is a recommendation dataset containing a mix of 27 categorical and 13 numerical features. Each row in the dataset is labeled with a binary class which we consider public. We also consider the numerical features to be public and keep the categorical features private. This choice of public and private features makes this task a very hard task for feature DP. This is because the categorical features alone are enough to get an accuracy as high as ~74\%, while with the mix of numerical and features we can get up-to ~75\%, which is a small improvement. We note that in our experiments we use a subset of 40,000 random samples from day 0 of the Criteo dataset for training and another 10,000 samples for validation/testing. 
    \item AFHQ is an animal face dataset. We use Gaussian blurring to blur each image and consider that to be a public version of the dataset (See Figure \ref{fig:afhq_samples}).
    \item LSUN is a dataset that contains images of scenes such as bedrooms, churches and bridges. Again, we use blurring for public feature (Figure \ref{fig:lsun_samples} in Appendix).
\end{itemize}

\subsection{Classification experiments}

\myparagraph{Training} 
\begin{figure}
  \centering
  \includegraphics[width=0.45\textwidth]{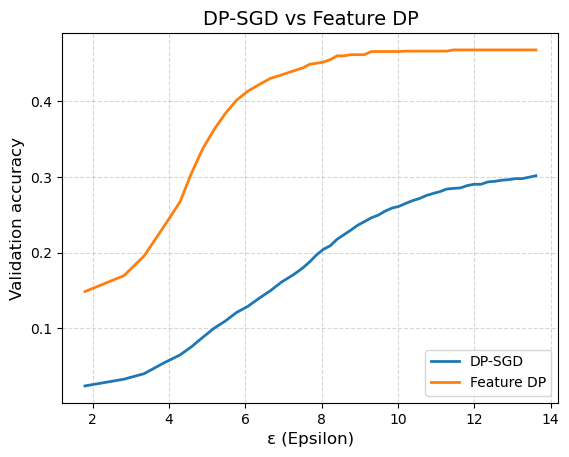}
  \caption{Purchase100: Feature DP vs DP}
  \label{fig:result}
\end{figure}
For the Purchase100 dataset, we use Algorithm \ref{alg:feature-DP} to train the model. We first create a public loss function as follows. Let $g: \R^{100} \to \R^{600}$ be a Gaussian padding that extends the 100 public features into a full vector of size $600$ and fills the private features with Gaussian noise $\mathcal{N}(0,1)$. We then define the public loss function  $\ell_{pub}(\mathbf{w}, \Psi(x),y)= \ell(\mathbf{w}, g(\Psi(x)), y)$ where $\ell$ is the cross entropy loss function. Then we set $\ell_{priv}(\mathbf{w}, x,y)=\ell(\mathbf{w},x,y) - \ell_{pub}(\mathbf{w},x,y)$. We use clipping of the gradient from the private loss function. In particular, we calculate $g_{priv}=\nabla \ell_{priv}\cdot \frac{\min(\lvert \nabla \ell_{priv}\rvert, C)}{\lvert\nabla \ell_{priv}\rvert}$ for $C=0.01$. Then we add Gaussian noise to get $\tilde{g}_{pub}= g_{pub} + \mathcal{N}(0,C^2\sigma^2)$. We also calculate $g_{pub}=\nabla \ell_{pub}$. Now, since clipping has biased the ratio between the norm of $g_{pub}$ and $\tilde{g}_{priv}$, we aggregate them using a ratio $\alpha$, $g=g_{pub}+ \alpha \tilde{g}_{priv}.$ We fix $\alpha$ in all iterations and tune it as a hyperparameter. We use a learning rate of $0.1$ and use momentum of $0.9$ to update our model.

For Criteo, we use a simple logistic regression model, with categorical features encoded using one-hot encoding. To enable this, we use sparse embeddings of dimension 1. To mask the private features (the categorical features), we set them to zero everywhere, which entirely eliminates the effect of categorical features. Specifically, assuming that the one-hot encoding of categorical features is $x_c$, the numerical features are $x_n$, the weight of the linear model for categorical features is $w_c$, the weight for numerical features is $w_n$, and $y$ is the label, the private loss becomes 

$$\mathbf{logistic}( w_c \times x_c + w_n \times x_n, y),$$ 

and the public loss becomes 

$$\mathbf{logistic}(w_n \times x_n, y).$$ 

Note that we could alternatively use a random value for each categorical feature to mask the true values, but we found that this leads to suboptimal performance. We also found that, for Criteo, performing 10 epochs of public pre-training (using only the public loss function) can significantly improve the performance. Interestingly, after this pre-training step, the best approach to achieve feature DP for Criteo is to perform a few steps of DP-SGD without using our formulation in Algorithm \ref{alg:feature-DP}. We believe this is because of the choice of public and private features . In particular, there is not much signal that can be extracted from the public features alone unless they are trained in conjunction with the private features. This also explains why the improvement from public features is lower in comparison with purchase dataset. Note that for Purchase dataset this does not hold. Namely, using public features for pretraining and then finetuning with all the features achieves a very limited improvement over DP-SGD. Specifically at epsilon values of 4.0 and 8.0, it achieves 9.2\% and 24.6\% respectively. This is in comparison with 24.2\% and 46.8\% for our algorithm.




  \begin{figure}
      \centering
      \includegraphics[width=1.0\linewidth]{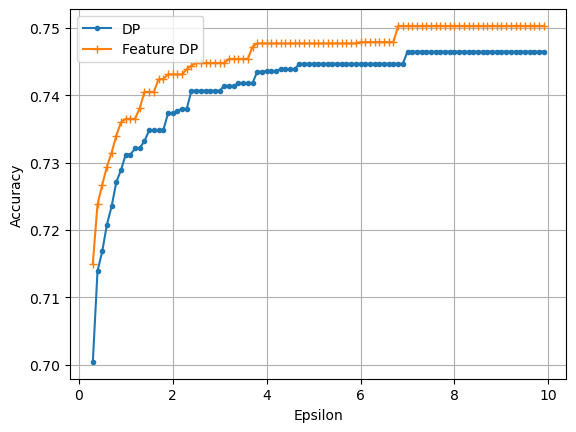}
      \caption{Comparing the privacy and utility trade-off for feature-DP and DP on Criteo dataset.}
      \label{fig:criteo}
  \end{figure}

\myparagraph{Results}
Figure \ref{fig:result} shows the result of our experiments for the Purchase100 dataset. As it is clear from the figure, for all values of $\epsilon$, we are able to improve the accuracy by $10-20\%$. This comes at the cost of leaking $100$ features from the set of $600$ features available. For both experiments in Figure \ref{fig:result}, the sampling rate is $1/16$ and the noise multiplier is set to $1.0$. The optimal learning rate schedule is different for two cases and is tuned for best results. 

Similarly, for the Criteo dataset, we observe that compared to DP, feature DP can improve the accuracy when leaking the numerical features. In the experiments, we tuned several hyperparameters. We used a range of clipping thresholds: [0.1, 1.0, 5.0, 10.0]. We considered batch sizes of [1,000, 2,000, 4,000, 8,000]. We also varied the noise multiplier from 1.0 to 6.0 and considered 1 to 20 epochs for the number of training steps. For each value of $\epsilon$, we first found the best set of hyperparameters for both DP and feature DP and then ran the experiments again with those hyperparameters. We repeated each experiment 5 times and reported the results in Figure \ref{fig:criteo}.




\begin{figure*}[ht]
    \centering
        \includegraphics[width=0.9\textwidth]{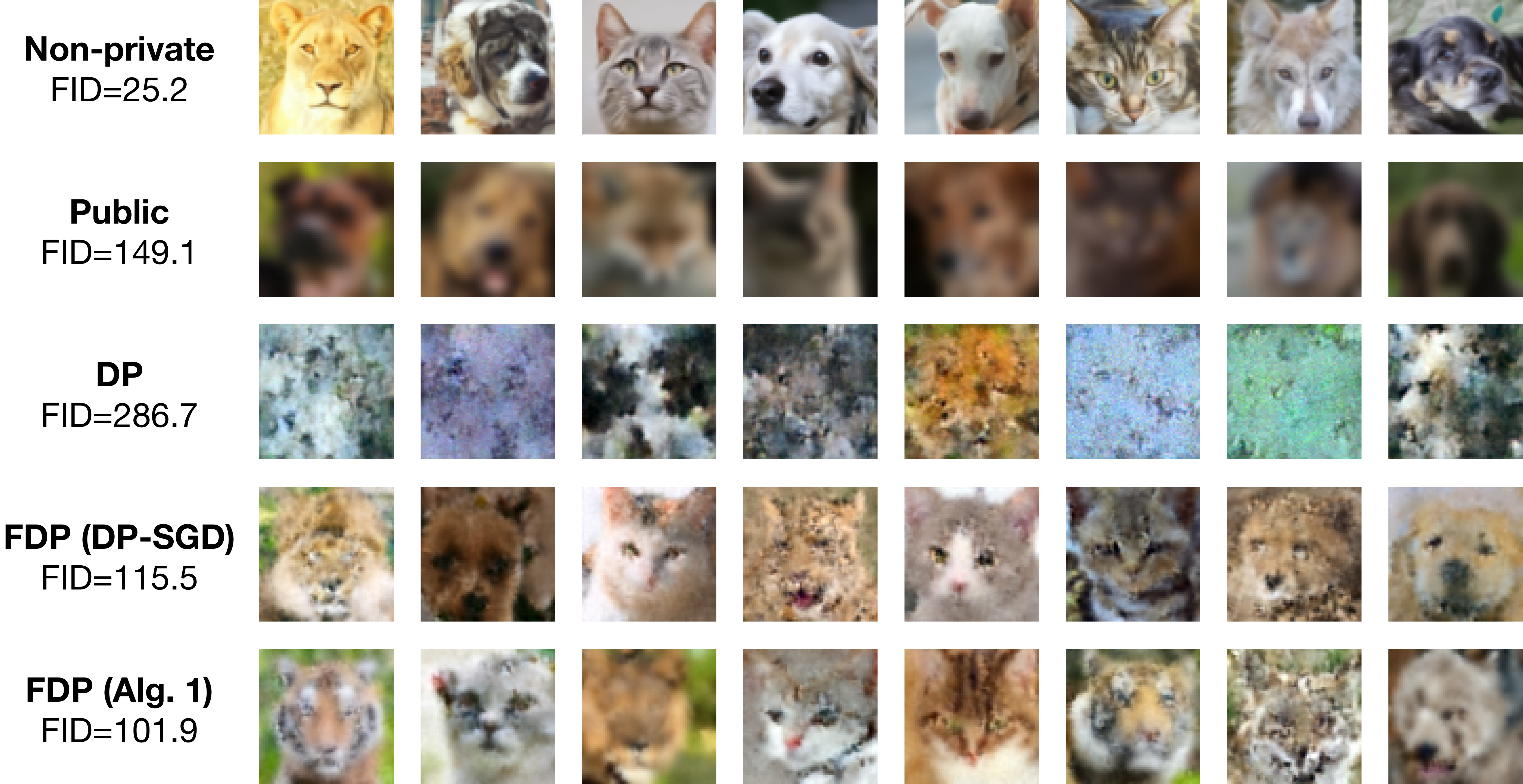}
    \caption{Generated image samples for AFHQ. FDP models trained using DP-SGD or Alg. \ref{alg:feature-DP} are able to produce clearly recognizable images of animal faces even under a strict FDP guarantee of $\epsilon=8$.}
    \label{fig:afhq_samples}
\end{figure*}
\subsection{Diffusion model experiments}
To demonstrate the versatility of our definition of feature DP, we apply it to the use case of training diffusion models for image generation. Specifically, given a private image, we consider its blurred version as the public feature. This form of blurring operation is commonly perceived to protect the privacy of images such as in face-blurring~\citep{yang2022study}. Thus, under our definition, training on the blurred version of an image satisfies $(0, 0)$-FDP with the blurring operation $\Psi$ as the feature map. Our goal is to train diffusion models (from scratch) that are capable of generating image samples that are of higher quality compared to the blurred images while satisfying rigorous FDP guarantees.



\myparagraph{Model and training} We adapt the diffusion model implementation from \url{https://github.com/VSehwag/minimal-diffusion}. The model has a U-Net architecture with a total of 32.8M parameters. We train a class-conditional model for the AFHQ dataset and an unconditional model for LSUN Bedroom. See Table \ref{tab:diffusion_hyperparameters} for details about hyperparameters.

\myparagraph{Methods} We consider the following methods/baselines.
\begin{itemize}[nosep, leftmargin=*]
    \item Non-private: This is the unmodified diffusion model training algorithm and serves as the reference point for non-private training.
    \item Public: This baseline trains on only the blurred images (i.e., public features) using the same non-private training algorithm, which is not DP but is $(\epsilon=0, \delta=0)$-FDP. We seek to outperform this baseline with FDP training.
    \item DP: This baseline performs standard DP-SGD training from scratch with Poisson subsampling and R\'{e}nyi DP accounting. By Theorem \ref{thm:FDP_composition}, this implies FDP with the same $\epsilon$ and $\delta$.
    \item FDP (DP-SGD) (\emph{ours}): We train the diffusion model using DP-SGD starting from the initialization obtained by training on blurred images. By Proposition \ref{prop:DP_to_FDP}, if the DP-SGD fine-tuning step is $(\epsilon,\delta)$-DP then the end-to-end training algorithm is $(\epsilon,\delta)$-FDP. 
    This algorithm bears similarity to the public pretraining algorithms used in the work of \citep{ghalebikesabi2023differentially, ganesh2023public} with two key differences. First, our baseline does not have an additional source of data. Rather, it uses a separation of public and private features on the training data. This separation enables us to perform two stages of training without needing any extra data.  Second, this algorithm does not satisfy differential privacy and it is only with our formalization of feature-DP that we can prove theoretical guarantees for it. 
    \item FDP (Alg. \ref{alg:feature-DP}) (\emph{ours}): We train the diffusion model using Algorithm \ref{alg:feature-DP} starting from the same initialization as in FDP (DP-SGD). We perform an additional post-processing step of rejection sampling to filter out blurry images.
\end{itemize}

For DP and feature DP, we target a privacy guarantee of $(\epsilon, \delta=1/2N)$-DP/FDP, where $N$ is the size of the training set, $\epsilon=8$ for AFHQ and $\epsilon=0.5$ for LSUN.
\begin{figure}
    \centering
    \begin{minipage}[c]{1.0\linewidth}
    \includegraphics[width=\linewidth]{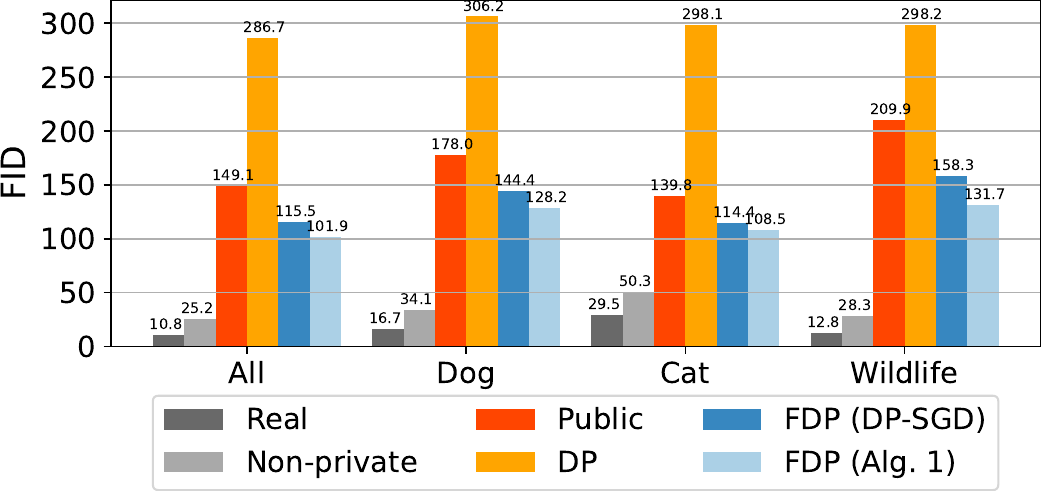}
        \caption{FID score evaluation for AFHQ; lower is better. The FID score is computed over $2,000$ random image samples vs. $1,500$ validation samples. We also show the per-class FID score for each of the dog, cat and wildlife classes.}
        \label{fig:afhq_fid}
    \end{minipage}%
\end{figure}
\begin{figure}[]
    \centering
    \begin{minipage}[c]{1.0\linewidth}
    \includegraphics[width=\linewidth]{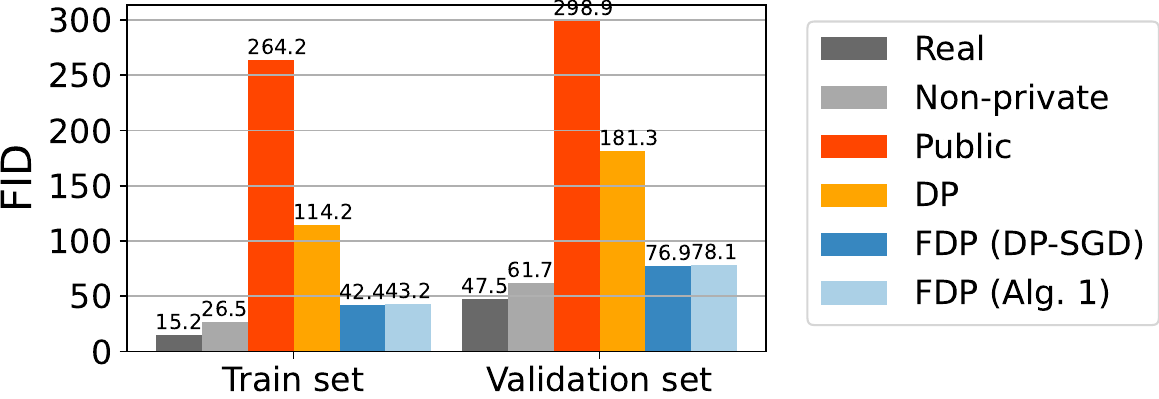}
    \caption{FID score evaluation for LSUN; lower is better. The FID score is computed over $2,000$ random image samples vs. $2,000$ train / $300$ validation samples.}
    \label{fig:lsun_fid}
    \end{minipage}
\end{figure}

\begin{table*}[t!]
    \centering
    \caption{Hyperparameters for diffusion model training.}
    \label{tab:diffusion_hyperparameters}
    \resizebox{0.9\textwidth}{!}{%
    \begin{tabular}{c|l|ccccc}
        \toprule
        Dataset & Method & \# Epochs & Batch size & Learning rate & Clipping factor $C$ & Noise multiplier $\sigma$ \\
        \midrule
        \multirow{2}{*}{AFHQ} & Non-private \& Public & 100 & 128 & $10^{-4}$ & - & - \\
                              & DP \& FDP & 100 & 128 & $10^{-4}$ & 0.1 & 0.86 \\
        \midrule
        \multirow{2}{*}{LSUN} & Non-private \& Public & 5 & 128 & $10^{-4}$ & - & - \\
                              & DP \& FDP & 500 & 16,384 & $10^{-4}$ & 0.1 & 15.6 \\
        \bottomrule
    \end{tabular}}
\end{table*}
\begin{figure*}[ht!]
    \centering
    \includegraphics[width=0.9\textwidth]{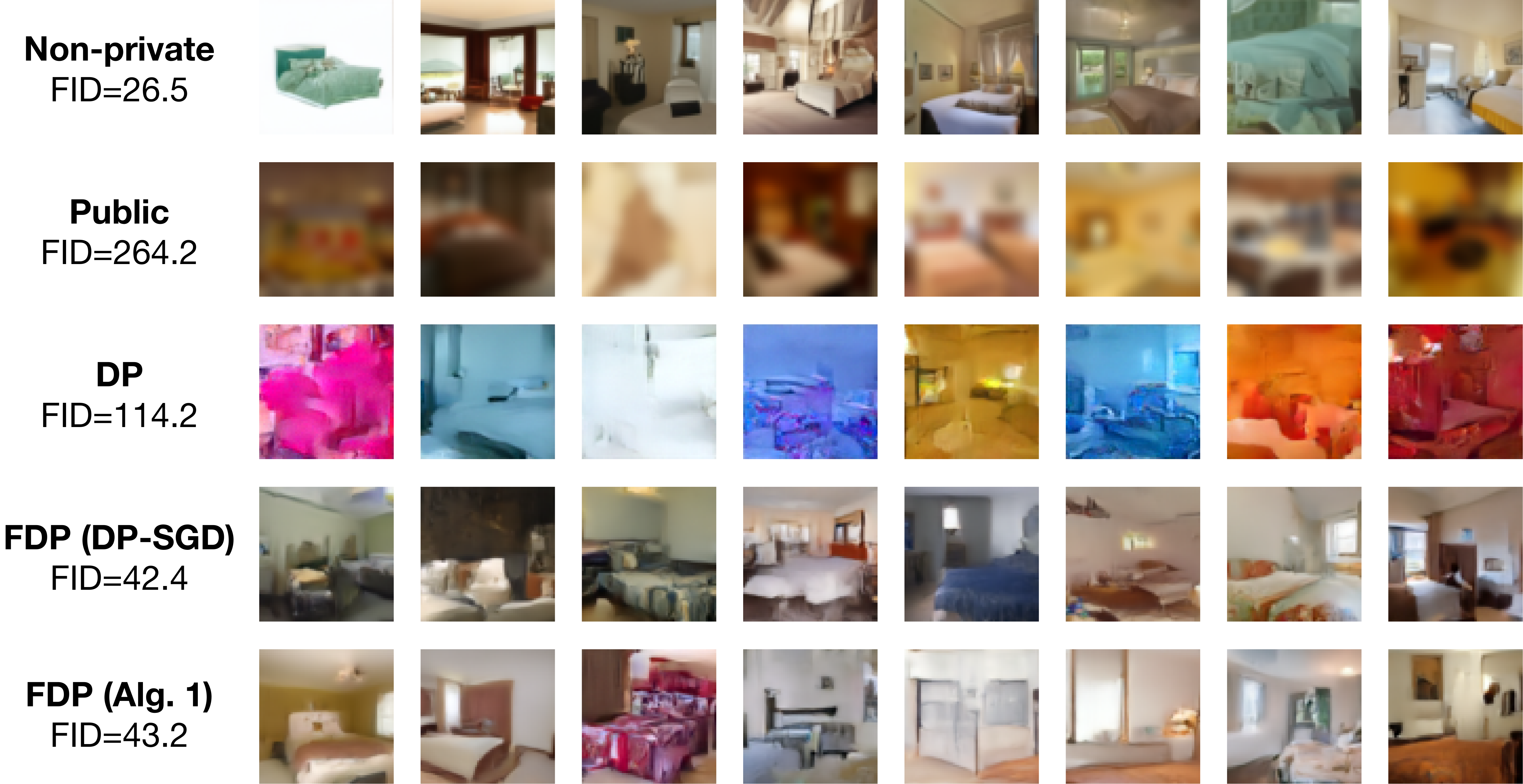}
    \caption{Generated image samples for LSUN. FDP (DP-SGD) and FDP (Alg. \ref{alg:feature-DP}) are capable of generating sharp images of bedrooms even under a strict privacy budget of $\epsilon=0.5$ due to having a large training set (3M) and batch size ($16,384$).}
    \label{fig:lsun_samples}
\end{figure*}\myparagraph{AFHQ result} Fig. \ref{fig:afhq_samples} shows samples generated from the diffusion model trained using various methods.
 The non-private model generates the highest quality animal faces that are clearly recognizable, attaining an FID score of $25.2$. As expected, the model trained on public features (blurred images) is only capable of generating blurred animal faces and attains a high FID score of $149.1$. The DP-trained model fails to generate any image that resembles an animal face, with an FID score of $286.7$ that is even higher than that of the public model. On the other hand, both FDP (DP-SGD) and FDP (Alg. \ref{alg:feature-DP}) perform surprisingly well, achieving FID scores of $115.5$ and $101.9$, respectively---both of which being below the FID score for the public model. The generated images are also clearly recognizable, although are much more pixelated compared to images generated by the non-private model.

We also present the full FID evaluation in Fig. \ref{fig:afhq_fid}. We sample $2,000$ images for each method and compute the FID score across the validation set. Since the model is class-conditional, we also compute the FID score for each class of dog, cat and wildlife independently. The quantitative result closely resembles the qualitative result in Fig. \ref{fig:afhq_samples}, with the non-private model having drastically lower FID compared to all other methods. Nevertheless, by starting from the initialization of the diffusion model trained on blurred images, FDP (DP-SGD) and FDP (Alg. \ref{alg:feature-DP}) are both capable of significantly reducing the FID score, outperforming DP training by a wide margin. Our result shows that feature DP can be a promising alternative definition that enables new classes of applications such as diffusion model training on privacy-sensitive images.\\

\myparagraph{LSUN result} We repeat the same evaluation for the diffusion models trained on LSUN bedroom. Figure \ref{fig:lsun_samples} shows the generated images samples and Fig. \ref{fig:lsun_fid} shows the FID score computed across $2,000$ random images from the training set
\footnote{We do this because the LSUN bedroom validation set is too small, so its FID score gives an inaccurate estimate of the distance between real and generated image distributions.}
, as well as the full validation set of $300$ images. Similar to the AFHQ experiment, the non-private model generates the highest quality images and attains the lowest FID score, while the public model generates blurry images with a high (train set) FID score of $264.2$. The DP-trained model under a strict privacy budget of $\epsilon=0.5$ is only capable of generating images with severe color distortion. Perhaps surprisingly, under the same privacy budget, FDP-trained models are able to generate sharp images of bedroom scenes consistently. Quantitatively, FDP (DP-SGD) and FDP (Alg. \ref{alg:feature-DP}) achieve FID scores of $42.4$ and $43.2$, respective, close to that of the non-private model's FID of $26.5$. We speculate that the relatively high performance of FDP-trained models is due to the large training set size (3M) and batch size ($16,384$),. This enabled the model to train for a large number of iterations under a small effective noise, starting from a good initialization provided by training on the blurred public images. Our observation is corroborated by prior results for discriminative~\citep{li2021large, de2022unlocking} and embedding model~\citep{yu2023vip} training, where scaling up the training set size and batch size proved to be important elements for the success of DP training.

\section{Discussion and Conclusion}
We have defined feature DP—a relaxation of DP that allows for the separation of a training sample into private and public features, providing rigorous protection for only the private ones. By leveraging this definition, we demonstrated that it is possible to achieve a substantially better privacy-utility trade-off for common ML tasks such as classification and generative modeling with diffusion models.

Our empirical results showed significant improvements in the quality of the models trained with feature-DP compared to traditional DP, which highlights the benefits of selectively applying privacy measures to sensitive components only. By distinguishing between private and public features, feature-DP not only ensures that privacy guarantees are upheld but also significantly enhances model utility, making privacy-preserving machine learning more practical for real-world applications. This is particularly useful in settings where non-sensitive data can contribute to learning without compromising the privacy of sensitive information.

We want to emphasize that, although feature-DP is presented as a \emph{relaxation} of the traditional notion of differential privacy, it can, in fact, lead to enhanced privacy outcomes. Due to the utility losses associated with strict privacy requirements, practitioners using traditional differentially private algorithms (e.g., DP-SGD) are often compelled to choose smaller noise levels or use loose parameters, ultimately resulting in weaker privacy guarantees. In contrast, our fine-grained definition of privacy enables the design of algorithms that focus solely on protecting the privacy of sensitive features, but with significantly stronger guarantees.

Moreover, our notion of feature-DP could improve privacy practices by encouraging practitioners who rely on ad-hoc and heuristic privacy protection methods to adopt differential privacy instead. By offering a more practical balance between privacy and utility, feature-DP makes it feasible to implement rigorous privacy measures where they are most needed, thereby fostering a more consistent use of formal privacy protections.


However, there are still some challenges and limitations that warrant further investigation. One limitation is the need for a robust method to accurately classify features as public or private, especially in dynamic or evolving datasets. Misclassification of sensitive features could lead to privacy breaches, whereas unnecessary privacy measures could degrade model performance. Future work could explore automated feature classification mechanisms, possibly using adversarial training or other machine learning techniques to better determine which features are critical to protect.

Another promising direction for future work is the development of specialized optimization algorithms tailored to feature-DP that fully exploit the separation between private and public features. Designing such algorithms could significantly enhance the utility of privacy-preserving machine learning models, making them more practical for deployment in industry applications. Moreover, further research into establishing theoretical bounds and guarantees for feature-DP across different learning scenarios would be invaluable. Such work could provide deeper insights into the trade-offs between privacy, utility, and computational efficiency, ultimately guiding practitioners in making more informed choices when balancing these crucial aspects.

Finally, we envision extending feature-DP to different modalities beyond image and tabular data, such as text and graph data. Exploring its applications in natural language processing, recommendation systems, and federated learning could open new avenues for improving privacy-preserving techniques across a wider range of machine learning domains. By focusing on adaptive privacy strategies, we believe that feature-DP can become an essential tool for bridging the gap between privacy and utility in machine learning.

\bibliographystyle{plain}
\bibliography{ref}


\appendices

\clearpage
\section*{Proofs}
\begin{proof}[Proof of Proposition \ref{prop:fdpi_to_FDPr}]
We prove this for $f(x)=e^\epsilon\cdot x + \delta$ and $fof(x)=e^{2\epsilon}x + \delta(1+e^\epsilon)$ (which corresponds to traditional differential privacy). The case of general $f$ follows similarly. We first prove the first side. Assume $M$ is \fDPi. Let $\cS$ and $\cS'$ be two datasets satisfying the replacement conditions of Definition \ref{def:fdp}. That is,
$$\cS\setminus \cS' = \set{x} \text{~~~and~~~} \cS'\setminus \cS = \set{x'} \text{~~~and~~~}\Psi(x)=\Psi(x')=u.$$
 By the fact that $M$ is \fDPi, we know that there is s simulator $\Sim$ such that for all $\cT$
$$ \Pr[\Sim(S\cap S', u)]\leq e^\epsilon\cdot\Pr[M(\cS')\in \cT]  + \delta $$
and 
$$\Pr[M(\cS)\in \cT]\leq e^\epsilon\cdot \Pr[\Sim(S\cap S'), u)] + \delta.$$
Combining the two inequalities we get
\begin{align*}\Pr[M(\cS)\in \cT]&\leq e^{2\epsilon}\Pr[M(\cS')\in \cT] + \delta(1+e^\epsilon)\\&= fof(\Pr[M(S')\in \cT]).
\end{align*}
This concludes the proof for first direction. For the other direction, we need to construct a simulator. Consider a simulator $\Sim$ that given $\cS, u$, first finds a point $x'$ such that $\Psi(x')=u$ and then runs the mechanism on $M(S\cup\set{x'})$. Note that there always exists an $x'$ because we know $\Psi(x)=u$. Now, let $\cS'=\cS \cup \set{x}$. We have
\begin{align*}
\Pr[\Sim(\cS,u)\in \cT] &= \Pr[M(\cS \cup \set{x'})\in \cT]\\
&\leq e^{\epsilon}\Pr[M(\cS')\in \cT] + \delta.
\end{align*}
which is followed by the fact that $M$ is \fDPr and similarly
\begin{align*}
\Pr[M(\cS')\in \cT]& \leq e^{\epsilon}\Pr[M(\cS \cup \set{x'})\in \cT]+ \delta. \\
&=e^{\epsilon}\Pr[\Sim(\cS,u)\in \cT]+ \delta.
\end{align*}\end{proof}

\begin{proof}[Proof of Proposition \ref{prop:DP_to_FDP}]
This directly follows from definition. A simulator can always ignore $u$ and run $M(S)$.
\end{proof}
\begin{proof}[Proof of Theorem \ref{thm:FDP_composition}]
    We prove this for $k=2$ and the rest follows by induction. Let $M_1\colon \cX^* \to \Theta$ and $M_2\colon \cX^*\times \Theta \to \Theta$. Note that we are allowing the mechanism $M_2$ to depend on an additional auxiliary input. This auxiliary input could be the output of $M_1$, so we can adaptively call $M_2$ followed by $M_1$. Let $\Sim_1\colon \cX^* \times \cF \to \Theta$ be the simulator for $M_1$  and $\Sim_2\colon \cX^* \times \cF \time \Theta \to \Theta$ be the simulators that satisfy feature DP for $M_2$. Again, note that $\Sim_2$ takes an additional parameter to allow for $\Sim_2$ to depend on all inputs to $M_2$. That is, $\Sim_2$ could be different for different auxiliary inputs to $M_2$. Now, we build a simulator $\Sim$ as follows: On an input dataset $\cS$ and a feature $u\in \cF$, $\Sim$ will first call $\Sim_1(\cS,u)$ to get a model $\theta$, then it calls $\Sim_2(\cS, u, \theta)$. We prove that $\Sim$ satisfies the notion of feature DP for $M_1 o M_2$.

    \begin{definition}[Hockey stick divergence and dominating pairs] The hockey stick divergence between two random variables $X$ and $Y$, (with pdfs $\mu_X$, $\mu_Y$) is defined as   
    $$HS_\alpha(X,Y) = \Ex_{\theta_\sim S}[(\frac{\mu_X(\theta)}{\mu_Y(\theta)}-\alpha)_{+}].$$ We say a pair of random variables $(X,Y)$ dominates another pair $(M,S)$ iff for all $\alpha\in \R$ 
    $$HS_\alpha(M,S))\leq HS_\alpha(X,Y).$$
    \end{definition}
    \begin{lemma}\label{lem:HS_to_fDP}
    Let $f=T(X,Y)$ be a trade-off function, for a pair of distributions $(M,S)$ we have 
    $$\forall \cT, \Pr_{\theta\sim M}[\theta \in \cT] \leq 1 - f(\Pr_{\theta\sim S}[\theta \in \cT])$$ if and only if $(M,S)$ is dominated by $(X,Y)$.
    \end{lemma}

    let $M\equiv (U_1,U_2)= (M_2(\cS \cup \set{x}), M_1 o M_2(\cS \cup \set{x})$ be the joint distribution of the result of running the mechanism.  Similarly $S\equiv(S_1,S_2)\equiv \Sim_1 o \Sim_2(\cS \cup{\Psi(x)})$. We use $S_2[\theta]$ and $U_2[\theta]$ to denote the conditional random variable of $S_2\mid S_1=\theta$ and $U_2\mid U_1=\theta$. Also let $\mu=(\mu_1,\mu_2)$ be the pdf of $M$ and $\nu=(\nu_1,\nu_2)$ the pdf of $S$. We also abuse the notation and $X_1(.)$, $X_2(.)$  $Y_1(.)$ and $Y_2(.)$ and the probability density functions for $X_1, X_2, Y_1,$ and $Y_2$. We have
    

    \begin{align*}
        &HS_\alpha (M, S) = \Ex_{\theta \sim S} [(\frac{\mu(\theta)}{\nu(\theta)}-\alpha)_+]\\
        &=\Ex_{\theta_1 \sim S_1}\Big[\Ex_{\theta_2 \sim S_2[\theta_1]}[(\frac{\mu(\theta_1,\theta_2)}{\nu(\theta_1,\theta_2)}-\alpha)_+]\Big]\\
        &=\Ex_{\theta_1 \sim S_1}\Big[\frac{\mu_1(\theta_1)}{\nu_1(\theta_1)}\Ex_{\theta_2 \sim S_2[\theta_1]}[(\frac{\mu_2[\theta_1](\theta_2)}{\nu_2[\theta_1](\theta_2)}-\alpha\cdot \frac{\nu_1(\theta_1)}{\mu_1(\theta_1)})_+]\Big]\\
        &= \Ex_{\theta_1 \sim S_1}\Big[\frac{\mu_1(\theta_1)}{\nu_1(\theta_1)}HS_{\alpha\cdot \frac{\nu_1(\theta_1)}{\mu_1(\theta_1)}}(M_2[\theta],S_2[\theta_2])\Big]\\
        &\leq \Ex_{\theta_1 \sim S_1}\Big[\frac{\mu_1(\theta_1)}{\nu_1(\theta_1)}HS_{\alpha\cdot \frac{\nu_1(\theta_1)}{\mu_1(\theta_1)}}(X_2,Y_2)\Big]\\
        &=\Ex_{\theta_1 \sim S_1}\Big[\Ex_{\theta_2 \sim Y_2}[(\frac{\mu_1(\theta_1)X_2(\theta_2)}{\nu_1(\theta_1)Y_2(\theta_2)}-\alpha)_+]\Big]\\
        &=\Ex_{\theta_2 \sim Y_2}\Big[\Ex_{\theta_1 \sim S_1}[(\frac{\mu_1(\theta_1)X_2(\theta_2)}{\nu_1(\theta_1)Y_2(\theta_2)}-\alpha)_+]\Big]\\
        &=\Ex_{\theta_2 \sim Y_2}\Big[\frac{X_2(\theta_2)}{Y_2(\theta_2)}\Ex_{\theta_1 \sim S_1}[(\frac{\mu_1(\theta_1)}{\nu_1(\theta_1)}-\alpha\frac{Y_2(\theta_2)}{X_2(\theta_2)})_+]\Big]\\
        &=\Ex_{\theta_2 \sim Y_2}\Big[\frac{X_2(\theta_2)}{Y_2(\theta_2)}HS_{\alpha\frac{Y_2(\theta_2)}{X_2(\theta_2)})}(M_1, S_1)\Big]\\
        &\leq \Ex_{\theta_2 \sim Y_2}\Big[\frac{X_2(\theta_2)}{Y_2(\theta_2)}HS_{\alpha\frac{Y_2(\theta_2)}{X_2(\theta_2)})}(X_1, Y_1)\Big]\\
         &=\Ex_{\theta_2 \sim Y_2}\Big[\Ex_{\theta_1 \sim X_1}[(\frac{X_1(\theta_1)X_2(\theta_2)}{Y_1(\theta_1)Y_2(\theta_2)}-\alpha)_+]\Big]\\
         &=HS_\alpha(X_1\times X_2, Y_1\times Y_2).
    \end{align*}
    With similar steps we can show that $H_\alpha(M,S) \leq HS_\alpha(X_1\times X_2, Y_1 \times Y_2).$ Therefore using Lemma \ref{lem:HS_to_fDP} we conclude the proof.
\end{proof}

Now we show that sub-sampling does not generally improve feature-DP. We construct a mechanism $M$ that is \fDPi with respect to a feature $\Psi$, but for any $p\in[0,1]$ its sub-sampled version does not satisfy \fDPi for any $\epsilon'< \epsilon$. 
\begin{proof}[Proof of Proposition \ref{prop:sub-samp-FDP}]
Let $\cB=\set{0,1}^d$ and define a mechanism $M\colon \cB^* \to \cB^*$ as follows. Let $\cS=\set{x_1, \dots, x_n; x_i\in \cB}$, we define $M(X)$
$$M(\cS) = \set{R(x_1), \dots,R(x_n)}$$
where $R(x_i)$ is Boolean random variable that is sampled based on the randomized response mechanism as follows:
\begin{align*}
R(x_i) = \begin{cases}
x_i^1,x_i^2,\dots, x_i^d &\text{With probability $\frac{e^{2\epsilon}}{1+e^{2\epsilon}}$}\\    
1-x^1_i,x_i^2,\dots, x_i^d &\text{With probability $\frac{1}{1+e^{2\epsilon}}$}
\end{cases}
\end{align*}
Also let us define the feature $\Psi\colon \cB \to \cB$ be a feature defined as
$$\Psi(b_1,\dots,b_d)= (b_2,\dots,b_d).$$ 

We first show that $M'=M~o~\poisson_p$ is \fDPi[$\epsilon$]. We construct a simulator $\Sim:\cB^*\times \cB \to \cB^*$ as follows:
\begin{align*}
\Sim\Big(\set{x_1,\dots, x_n}, \Psi(x_{n+1})\Big) &=M'(\set{x_1,\dots, x_n})\\&~~~\cup G_p(\Psi(x_{n+1})).
\end{align*}
where 
\begin{align*}
G_p(b_2,\dots,b_d)=
\begin{cases}
    \emptyset& \text{With probability $1-p\frac{e^\epsilon}{1+e^{2\epsilon}}$},\\
    \set{(0,b_2,\dots,b_d)}& \text{With probability $\frac{p}{2}$},\\
    \set{(1,b_2,\dots,b_d)}& \text{With probability $\frac{p}{2}$}.
\end{cases}
\end{align*}
In what follows we use $G$ to denote $G_{1.0}$.
Now we use $\Sim$ as a simulator for $M' = M o~\poisson_p$. 
Consider two datasets $\cS=\set{x_1,\dots,x_n}$ and $\cS'=\cS \cup \set{x_{n+1}}$.

Let $\cT^*\subset \cB$ We have
\begin{align*}\Pr[M'(\cS') =\cT] &= (1-p)\cdot \Pr[M'(\cS) =\cT]\\ 
&~~~+ p \cdot \Pr[M(\cS) \cup \set{R(x_{n+1})} = \cT].
\end{align*}
On the other hand, we have
\begin{align*}\Pr[\Sim(\cS, \Psi(x_{n+1})) =\cT] &= (1-p)\cdot \Pr[M'(\cS) =\cT]\\ 
&+ p \cdot \Pr[M(\cS) \cup \set{G(\Psi(x_{n+1}))} = \cT].
\end{align*}
We now only focus on the second terms. Let $b$ be the first bit of $x_{n+1}$. We have
\begin{align*}q_1&=\Pr[M(\cS) \cup \set{R(x_{n+1})} = \cT]\\  &=
(1-\frac{1}{2e^\epsilon}) \Pr[M(\cS) \cup \set{(b,\Psi(x_{n+1)})} = \cT]\\
&+\frac{1}{2e^\epsilon} \Pr[M(\cS) \cup \set{(1-b,\Psi(x_{n+1)})} = \cT].
\end{align*}
We also have
\begin{align*}q_2&=\Pr[M(\cS) \cup \set{G(\Psi(x_{n+1}))} = \cT]\\  &=
\frac{1}{2} \Pr[M(\cS) \cup \set{(b,\Psi(x_{n+1}))} = \cT]\\
&+\frac{1}{2} \Pr[M(\cS) \cup \set{(1-b,\Psi(x_{n+1}))} = \cT].
\end{align*}
Let $q_3 = (1-\frac{1}{2e^\epsilon}) \Pr[M(\cS) \cup \set{(b,\Psi(x_{n+1}))} = \cT]$, $q_4=\frac{1}{2e^\epsilon} \Pr[M(\cS) \cup \set{(1-b,\Psi(x_{n+1}))} = \cT]$, $q_5 = \frac{1}{2} \Pr[M(\cS) \cup \set{(b,\Psi(x_{n+1}))} = \cT]$, $q_6=\frac{1}{2} \Pr[M(\cS) \cup \set{(1-b,\Psi(x_{n+1}))} = \cT]$ so that we have $q_1$=$q_3 + q_4$ and $q_2=q_5+q_6$. We can now check that 

$$e^{-\epsilon}\leq \frac{q_3}{q_5}\leq e^\epsilon\text{~~and~~}e^{-\epsilon}\leq \frac{q_4}{q_6} \leq e^\epsilon.$$
Therefore we have $$e^{-\epsilon}\leq\frac{(1-p)\Pr[M'(\cS) =\cT] + p\cdot q_1}{(1-p)\Pr[M'(\cS) =\cT] + p\cdot q_2}\leq e^\epsilon.$$

Now we prove the tightness. We need to show that for any simulator $\Sim$ there is a pair of datasets and a possible outcome where the ratio of probabilities is at least $e^\epsilon$. Let $\cS=\set{\mathbf{0}}$ and $\cS'=\set{\mathbf{a},\mathbf{b}}$ and $\cS''=\set{\mathbf{a},\mathbf{c}}$ where 
$$\mathbf{a}=(0,\dots,0), \mathbf{b}=(1,\dots,1) \text{~and~} \mathbf{c}=(0,1,\dots,1) .$$
Now let $\Sim$ be an arbitrary simulator that satisfies \fDPi{$(\epsilon',0)$} for $\epsilon'<\epsilon$. Let $\cT_1$ be the set of all possible outcomes with $\mathbf{b}\in \cT_1$.

We have, 
\begin{align*}
q_1=\Pr[M'(\cS')\in \cT_1] = p\cdot (1-\frac{1}{2e^\epsilon}).
\end{align*}
and
\begin{align*}
q_2=\Pr[M'(\cS'')\in \cT_1] = \frac{p}{2e^\epsilon}.
\end{align*}
Now let us define
\begin{align*}
q_3=\Pr[\Sim(\cS, \Psi(\mathbf{b})) \in \cT_1] = \Pr[\Sim(\cS, \Psi(\mathbf{c}))\in \cT_1].
\end{align*}
\begin{align*}
q_4=\Pr[M'(\cS')\not \in  \cT_2 \cup \cT_1] = 1-p.
\end{align*}
\end{proof}

\begin{proof}[Proof of Theorem \ref{thm:att_inf}]
We want to bound the following quantity.
    $$adv(A, X, u, S, d, \rho)= \Pr_{\substack{x\sim [X\mid \Psi(X)=u]\\ \theta \sim M(S \cup \set{x})\\
    x' \sim A(\theta)}}[d(x', x)\leq \rho].$$
For simplicity, let us assume $A$ is deterministic. Note that this assumption is without loss of generality because we can always fix the randomness of $A$ to be the best possible randomness for the objective, and by averaging argument the advantage of such deterministic attack is at least as good as the randomized attack. 
We rewrite this expectation as follows: 

\begin{align*}
    &\Pr_{\substack{x\sim [X\mid \Psi(X)=u]\\ \theta \sim M(S \cup \set{x})\\
    x' \sim A(\theta)}}[d(x', x)\leq \rho]\\
    &~~~= \E_{x\sim [X\mid \Psi(X)=u]}[\Pr_{\theta \sim M(S \cup \set{x})}[(A(\theta),x)\leq \rho]].
\end{align*}
Let us define an event $E$ on the model parameters where $E_x=\set{\theta\in \Theta; d(A(\theta),x)\leq \rho}$.
Therefore we have,
\begin{align*}
    &\Pr_{\substack{x\sim [X\mid \Psi(X)=u]\\ \theta \sim M(S \cup \set{x})\\
    x' \sim A(\theta)}}[d(x', x)\leq \rho]\\
    &~~~~= \E_{x\sim [X\mid \Psi(X)=u]}[M(S\cup \set{x})\in E_x ].
\end{align*}
let $\Sim$ be the simulator for feature $\Psi$ for which the feature DP is satisfied. By the definition of feature DP we have, 

$$\Pr[M(S\cup \set{x}) \in E_x] \leq 1-f\big(\Pr[\Sim(S, u)\in E_x]) \leq \rho]\big) .$$
Therefore using Jensen inequality and the definition of $\ball(\cdot)$ we have,
\begin{align*}&\E_{x\sim [X\mid \Psi(X)=u]}[M(S\cup \set{x}) \in E_x]\\
&~~~\leq 1- \E_{x\sim [X\mid \Psi(X)=u]}[f(\Pr[\Sim(\cS, u)\in E_x])]
\\
&~~~\leq 1- f(\E_{x\sim [X\mid \Psi(X)=u]}[\Pr[\Sim(\cS, u)\in E_x]])
\end{align*}
On the other hand we have
\begin{align*}&\E_{x\sim [X\mid \Psi(X)=u]}[\Pr[\Sim(\cS, u)\in E_x]]\\
&~~~= \E_{x\sim [X\mid \Psi(X)=u]}[\E_{\theta\sim A(\Sim(\cS, u))}[d(A(\theta),x)\leq \rho]]\\
&~~~=\E_{\theta\sim A(\Sim(\cS, u))}[\E_{x\sim [X\mid \Psi(X)=u]}[d(A(\theta),x)\leq \rho]]\\
&~~~\leq \E_{\theta\sim A(\Sim(\cS, u))}[\ball(X,u,d,\rho)]].\\
&~~~=\ball(X,u,d,\rho)
\end{align*}
\end{proof}

\begin{proof}[Proof of Theorem \ref{thm:utility}]First note that the gradient  $g_t$ is an unbiased estimate of the empirical gradient by linearity of expectation. We can also bound the expected $\ell_2$ norm of the gradient by $\tau'^2$ and the variance of the added noise by $d\sigma^2$. Therefore, applying the standard stochastic gradient oracle techniques (Theorems 2 and 3 in \citep{shamir2013stochastic}  for analyzing the convergence of SGD, we can bound the gap between the empirical risk of the optimal and the obtained model to be at most
\[(\frac{M^2}{c} + c\tau'^2 + cd\sigma^2)\frac{2+\log(T)}{\sqrt{T}}.\] in the case of convex loss functions (See Theorem 2 in \citep{shamir2013stochastic}). Similarly, for the $\lambda$- strongly convex losses, by using Theorem 3 in \citep{shamir2013stochastic} we can bound the excess empirical risk by

$$17(\tau'^2 + d\sigma^2)\frac{1+\log(T)}{\lambda T}.$$
\end{proof}

\begin{proof}[Proof of Theorem \ref{thm:privacy}]
Here we first analyze one step of the mechanism. Let us start with a lemma about feature-DP. For warm-up we first do the privacy analysis for the case of privacy with replacement, then we show how the privacy with insertion deletion follows similarly by defining a proper simulator. 
\begin{proposition}\label{propone}
Let $\Psi$ be a feature and $M_\lambda$ be a parameterized mechanism that is $f$ feature DP with respect $\Psi$, (\fDPr[f]), for all parameters $\lambda$. Also, let $M'_\lambda$ be an arbitrary parameterized mechanism that operates on $\Psi(\cS)$. Then, adaptive composition of $M$ and $M'$ in both orders ($MoM'$ and $M'oM$) is also $f$ feature DP (\fDPr[f]). Note that adaptive composition means that the parameter of $M$ (or $M'$) could be a function of the output of $M'$ (or $M$). 
\end{proposition}
\begin{proof} Let us use $\Lambda(\theta)$ to denote the parameter of $\theta$ that is passed to the parameterized mechanisms $M$ and $M'$. We have
\begin{align*}
\Pr[MoM'(\cS) \in S] &= \E_{\lambda \sim \Lambda(M'(\Psi(\cS))}\Pr[M_\lambda(\cS) \in S]\\&= \E_{\lambda \sim \Lambda(M'(\Psi(\cS'))}\Pr[M_\lambda(\cS) \in S]\\ &\leq 
\E_{\lambda \sim \Lambda(M'(\Psi(\cS'))}  [1-f(\Pr[M_\lambda(\cS') \in S)]\\
&\leq 1- f(\E_{\lambda \sim \Lambda(M'(\Psi(\cS'))}  [\Pr[M_\lambda(\cS') \in S])\\
&=1- f(  [\Pr[MoM'(\cS') \in S])
\end{align*}
Note that the last two lines follow by convexity of $f$.
Similarly, for the other direction we have,
\begin{align*}
\Pr[M'oM(\cS) \in S] &= \E_{\lambda \sim \Lambda(M(\cS))}\Pr[M'_\lambda(\Psi(\cS)) \in S]\\&= \E_{\lambda \sim \Lambda(M(\cS))}\Pr[M'_\lambda(\Psi(\cS')) \in S]\\ &\leq 1- f(\E_{\lambda \sim \Lambda(M(\cS'))}  \Pr[M'_\lambda(\Psi(\cS')) \in S])\\
&= 1-f(\Pr[M'oM(\cS') \in S]).
\end{align*}
\end{proof}
Note that in the Proposition above, it is crucial that $M$ and $M'$ do not share internal randomness. This is the main reason we cannot use the same batch for the public and private loss functions.

Now note that each iteration of the algorithm can be stated as $MoM'$ where $M'(\Psi(\cS))$ is the process of calculating $g_t^{pub}$, and $M(\cS)$ is a sub-sampled Gaussian mechanism that calculates $\frac{1}{m}\sum_{x\in D} \nabla \ell_{priv}(\mathbf{w}_t, x) + \mathcal{N}(g_t^{pub},\sigma^2).$
Therefore, each step of the algorithm is as private as of a sub-sampled Gaussian mechanism with sub-sampling rate $q=m/n$ and noise multiplier $\sigma/\tau$. So, using Propositions \ref{propone} and \ref{prop:DP_to_FDP} also applying the composition theorem \ref{thm:FDP_composition} we finish the proof. Additionally, by converting $f$-DP to DP we get that the mechanism is $(\epsilon,\delta)$-DP  for $\epsilon=c\frac{m\tau}{\sigma\cdot n} \sqrt{\log(\frac{1}{\delta}) \log(\frac{T}{\delta})}$ and  some constant $c$.

For the case of insertion/deletion, we need to design a simulator. Consider a simulator $sim$ that is defined by Algorithm \ref{alg:feature-DP-sim}. This algorithm takes the same inputs as of algorithm \ref{alg:feature-DP}, but it also takes an additional public feature value $u$. We have highlighted the only difference between the simulator and Algorithm $\ref{alg:feature-DP}$ in color red.

Now let $X_i$ be a random variable that corresponds to $g_i$ in algorithm \ref{alg:feature-DP-sim} and $Y_i$ be the the random variable that corresponds to $g_i$ in Algorithm \ref{alg:feature-DP}. By construction of the simulator, we know that $X_i$ and $Y_i$ are both mixture of Gaussian distributions $\cN(A_i, \sigma^2)$, $\cN(B_i, \sigma^2)$ where $B_i$ itself is a mixture of $A_i$ and $C_i=A_i+ \nabla \ell_{priv}(\mathbf{w}_t; x)$, with $B_i=(1-p)\cdot A_i + p \cdot C_i$. Therefore, the hockey-stick divergence between $X_i$ and $Y_i$ is bounded by that of $\cN(0, \sigma^2)$ and $(1-p)\cN(0, \sigma^2) + p\cN(0,\sigma^2)$ (See Theorem 4 in~\citep{balle2018privacy}). Therefore, for each step of the mechanism we have $T(X_i, Y_i)$ is dominated by $T\Big(\cN(0,\sigma^2), (1-p)\cN(0,\sigma^2) + p\cN(L, \sigma^2)\Big)$ and therefore by composition of trade-off functions (See Theorem 3.2 in \citep{dong2019gaussian}) for the entire mechanism we have 
$$T\Big((X_1,\dots, X_T),(Y_1,\dots, Y_T)\Big)$$ is dominated by $$f=T\Big(\cN(0,\sigma^2)^T , \big((1-p)\cN(0,\sigma^2) + p\cN(L, \sigma^2)\big)^T\Big).$$

This implies that for any event $\cT$, setting $X=(X_1,\dots, X_n)$ and $Y=(Y_1,\dots, Y_n)$ we have
$$\Pr[X \in \cT] \leq 1-f(\Pr[Y\in \cT]).$$
And this shows that the mechanism satisfies $f$-DP with respect to feature $\Psi$ for the function $f$ specified in Theorem \ref{thm:privacy}.

\begin{algorithm}
\caption{Simulator for Noisy SGD with Public Features}\label{alg:feature-DP-sim}
\begin{algorithmic}[1]
\REQUIRE Public feature $\Psi$, Dataset $\cS$, Batch sizes $m$, $m'$, Learning rate schedule $\eta(t)$, standard deviation $\sigma$, Projection space $\mathcal{W} \in R^d$, Loss functions $\ell_{priv}, \ell_{pub}$, Number of iterations $T$, {\color{red} public feature value $u$}
\STATE Initialize $\mathbf{w}_1 \in \mathcal{W}$
\FOR{$t = 1, \dots, T$}
    \STATE Sample a mini-batch $B^{priv}_t$ with Poisson sampling with probability $p=m/|\cS|$. 
    \STATE $g^{priv}_{t}~=~\frac{1}{m}\sum_{x\in B_t^{priv}}\nabla \ell_{priv}(\mathbf{w}_t; x)$
    \STATE {\color{red} Construct a set $\cS''$ = $\cS \cup \set{x'}$ with $\Psi(x') = u$  }
    \STATE{\color{red} Sample a second mini-batch $B^{pub}_t$ of size $m'$ uniformly at random from $\cS''$.}
    
    \STATE $g_t^{pub}=~\frac{1}{m'}\sum_{x\in B_t^{pub}}\nabla \ell_{pub}(\mathbf{w}_t; \Psi(x)).$
    \STATE Let $g_t~=~g_t^{pub} + g_t^{priv} + \mathcal{N}(0, \sigma^2)$
    \STATE Update $\mathbf{w}_{t+1} = \mathbf{w}_t - \eta(t)\cdot g_t$
    \STATE Project $\mathbf{w}_{t+1}$ into the set $\mathcal{W}$: $\mathbf{w}_{t+1} = \text{Proj}_{\mathcal{W}}(\mathbf{w}_{t+1})$
\ENDFOR
\STATE \textbf{return} $\mathsf{aggregate}(\mathbf{w}_1,\dots,\mathbf{w}_{T+1})$
\end{algorithmic}
\end{algorithm}
\end{proof}

\newpage
\section{Meta-Review}

The following meta-review was prepared by the program committee for the 2025
IEEE Symposium on Security and Privacy (S\&P) as part of the review process as
detailed in the call for papers.

\subsection{Summary}
The paper describes feature differential privacy (FDP), a relaxation variant of DP to only protect certain features in individual data records.  FDP utilizes a secondary loss function, together with two separate random batch selection, to achieve formal privacy guarantee, especially against attribute inference attacks.

\subsection{Scientific Contributions}
\begin{itemize}
\item Addresses a Long-Known Issue
\item Provides a Valuable Step Forward in an Established Field
\end{itemize}

\subsection{Reasons for Acceptance}
\begin{enumerate}
\item The notion of “Feature Differential Privacy” is novel, which protects against attribute inference attacks.
    \item This paper extends DP-SGD algorithm to satisfy feature DP, and shows that feature DP can achieve better utility guarantees than approximate DP.
    \item This paper allows one to selectively decide which features to protect.

\end{enumerate}

\subsection{Noteworthy Concerns} 
\begin{enumerate} 
\item Accuracy gains on the Criteo dataset is relatively small.
\end{enumerate}

\end{document}